\documentclass[10pt,final] {IEEEtran}

\usepackage{graphicx}
\usepackage{amsmath,amssymb, amsfonts}
\usepackage{amsthm}
\usepackage{epstopdf}
\usepackage{multirow}
\usepackage{psfrag}
\usepackage[noadjust]{cite}
\usepackage{subfigure}

\usepackage{booktabs}
\usepackage{bbold}
\usepackage[ruled,vlined]{algorithm2e}
\usepackage{color}
\usepackage{mathtools}
\usepackage{fancyhdr}
\usepackage{subfigure}

\newcommand{\real}{\mathbb{R}}

\newtheorem{theorem}{Theorem}
\newtheorem{remark}{Remark}
\newtheorem{assn}{Assumption}
\newtheorem{proposition}{Proposition}
\newtheorem{definition}{Definition}
\newtheorem{convention}{Convention}

\newcommand{\until}[1]{\{1,\dots, #1\}}
\newcommand{\subscr}[2]{#1_{\textup{#2}}}
\newcommand{\supscr}[2]{#1^{\textup{#2}}}
\newcommand{\map}[3]{#1: #2 \rightarrow #3}

\begin{document}
\title{Asynchronous and Dynamic Coverage Control Scheme for Persistent Surveillance Missions\thanks{
This work has been sponsored by
the U.S.\ Army Research Office and the Regents of the University of
California, through Contract Number W911NF-09-D-0001 for the Institute for
Collaborative Biotechnologies, and that the content of the information does
not necessarily reflect the position or the policy of the Government or the
Regents of the University of California, and no official endorsement should
be inferred.}} 
\author{Jeffrey R. Peters, Sean J. Wang\thanks{Jeffrey R. Peters and Sean J. Wang are with the Mechanical Engineering Department and the Center for Control, Dynamical Systems and Computation, University of California, Santa Barbara \texttt{jrpeters@engr.ucsb.edu}, \texttt{seanwang@umail.ucsb.edu}}, Amit Surana\thanks{Amit Surana is with the Systems Department at United Technologies Research Center, East Hartford, CT, \texttt{SuranaA@utrc.utc.com}}, and Francesco Bullo\thanks{Francesco Bullo is with the Mechanical Engineering Department and the Center for Control, Dynamical Systems and Computation, University of California, Santa Barbara \texttt{bullo@engr.ucsb.edu}}}
\maketitle
\begin{abstract}
A \emph{decomposition-based} coverage control scheme is proposed for multi-agent, persistent surveillance missions operating in a communication-constrained, dynamic environment. 
The proposed approach decouples high-level task assignment from low-level motion planning in a modular framework. 
Coverage assignments and surveillance parameters are managed by a central base station, and transmitted to mobile agents via unplanned and asynchronous exchanges. Coverage updates promote load balancing, while maintaining geometric and temporal characteristics that allow effective pairing with generic path planners. Namely, the proposed scheme guarantees that (i) coverage regions are connected and collectively cover the environment, (ii) subregions may only go uncovered for bounded periods of time, (iii) collisions (or sensing overlaps) are inherently avoided, and (iv) under static event likelihoods, the collective coverage regions converge to a Pareto-optimal configuration. This management scheme is then paired with a generic path planner satisfying loose assumptions. The scheme is illustrated through simulated surveillance missions.
\end{abstract}

\section{Introduction}
\label{sec:introduction}
\subsection{Decomposition-Based Multi-Agent Surveillance}
Modern exploratory and surveillance missions often utilize autonomous vehicles or sensors to observe and monitor
large geographic areas. Example application domains include search and rescue~\cite{AM-GN-BB:11}, environmental monitoring~\cite{RNS-YC-PPL-DAC-BHJ-GSS:10}, warehouse logistics~\cite{PRW-RdA-MM:08}, and military reconnaissance \cite{SRD-CDW:06}. Such scenarios require robust and flexible tools for autonomous coordination.

A number of planning strategies have been studied for single agents, ranging from simple \emph{a priori} tour construction~\cite{FP-JWD-FB:11h} to more complex methods involving Markov chains~\cite{VS-FP-FB:11za}, optimization~\cite{NM-SLS-SLW:15}, or Fourier analysis~\cite{GM-AS-IM:10}. Unfortunately, it is not straightforward to generalize single-agent strategies for use in multi-agent missions: Naive approaches where each agent follows an independent policy can result in poor performance and introduce collision risks, while sophisticated generalizations often require joint optimizations that are intractable for even modestly sized problems. Scaling issues are sometimes alleviated through distributed control; however, such setups are application specific and may require extensive efforts to pose a mathematical problem that is suitable for use with formal techniques~\cite{NN:14}. In contrast, \emph{decomposition-based} approaches, which decouple the assignment and routing problem by first dividing the workspace among agents, offer a straightforward, modular framework to reasonably accomplish the desired goals, despite sacrificing optimality in general.



Communication constraints can, however, make it difficult to effectively divide a dynamic workspace in real time, while still ensuring that the result is amenable for pairing with single-agent motion planners. For example, many unmanned missions require agents to transfer sensor data to a central base station for analysis. With difficult terrain or hardware limitations, these sporadic exchanges with the base station may provide the only means of sharing real-time information across agents. Applications that operate under this constraint include underwater gliders that must surface to communicate with a tower~\cite{AP-HH-GS:10}, data mules that periodically visit ground robots~\cite{RCS-SR-WB:03}, and supervisory missions where human operators analyze data collected by unmanned vehicles~\cite{JP-VS-AS-GT-AS-MPE-FB:12t}. Here, updated mission information can only be relayed to one agent at a time, rendering traditional dynamic partitioning schemes, which rely on complete or pairwise coverage updates, impossible. As such, designers are forced to either allow for overlap in coverage assignments, or sacrifice complete coverage (Fig.~\ref{fig:motivation}).
\begin{figure}
\centering
\includegraphics[width = 0.8\columnwidth]{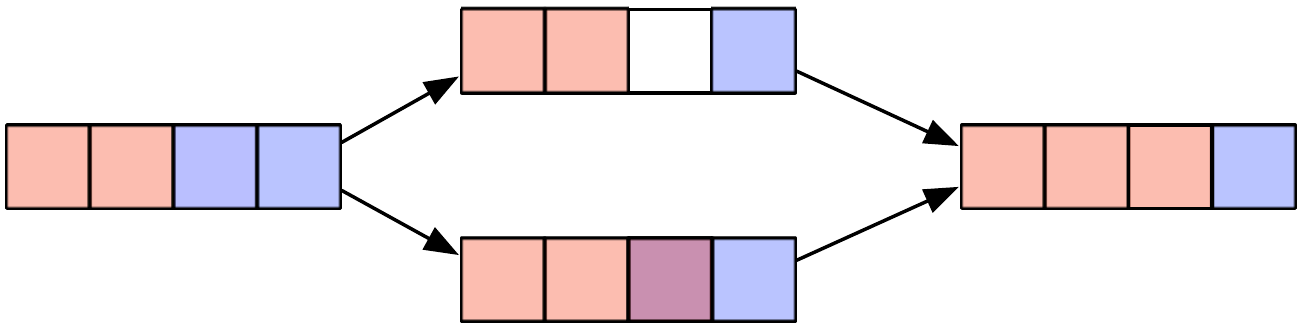}
\caption{Difficulties arise when complete or pairwise updates are impossible. With only single-region updates, two steps are required to move from the left-most to the right-most configuration. In the top path, blue is updated first, leaving an uncovered block at the intermediate step. In the bottom path, red is updated first, leaving a block that belongs to both regions simultaneously.}
\label{fig:motivation}
\end{figure}

This work presents a decomposition-based surveillance framework that operates under asynchronous \emph{one-to-base station} communication, a protocol in which data is transferred solely via sporadic exchanges between a central base station and autonomous agents~\cite{RP-PF-JWD-RC-FB:13w}. 
The complete scheme consists of two components: a dynamic partitioning component and a single-agent routing component. The work herein focuses primarily on developing the partitioning component, which, in addition to dividing the workload among the agents, manipulates local variables in order to allow for effective pairing with single-agent route planners. 
This development naturally leads to a complete multi-agent framework in which high-level coverage is coordinated by the base station, while individual trajectories are generated independently via on-board planners.  

\subsection{Related Literature}
\label{sec:related}
Research relating to multi-agent coverage control is vast. Typical strategies involve optimization~\cite{JP-SA:05}, auctions~\cite{BPG-KJM:02}, biological meta-heuristics~\cite{EB-MD-GT:99},  potential fields~\cite{JHR-HW:99}, or space decomposition~\cite{NN-IK:08}. Of particular relevance are multi-agent \emph{persistent surveillance} (\emph{persistent monitoring}) problems, in which a mobile sensor team is tasked with continual surveillance of a region of interest, requiring each subregion to be visited multiple (or infinitely many) times with the goal of minimizing a cost, e.g., the time between visits or the likelihood of detecting stochastic events~\cite{NN:14}.  Persistent surveillance is a generalization of \emph{patrolling}, where closed tours are sought for the purpose of protecting or supervising an environment. Patrolling has been studied extensively, with most current solutions being based on operations research algorithms, non-learning multi-agent systems, and multi-agent learning~\cite{AA-GC-HS-PT-TM-VC-YC:04}. Patrolling formulations are often one-dimensional and solutions usually reduce to ``back and forth'' motions that do not readily extend to general scenarios (see, e.g.~\cite{FP-FZ-JRP-MS-RC-FB:12l}).

The framework proposed herein is based on \emph{workspace decomposition}, and uses an allocation strategy to reduce the multi-agent problem into a set of single-agent problems. This approach is common in multi-agent systems due to its simplicity and scalability~\cite{NN:14}. For persistent surveillance of planar regions, decomposition-based approaches consist of two primary components: partitioning and single-agent routing. The most common approaches to optimal partitioning are based on Voronoi partitions~\cite{AO-BB-KS-SNC:00}. Effective schemes exist for constructing centroidal Voronoi, equitable, or other types of optimal partitions under various communication, sensing, and workload constraints~\cite{JC-SM-TK-FB:02j, RP-PF-RB:13i,JC:10}. 
These strategies usually require convex environments. More general workspaces are usually addressed by discretizing the environment and considering the resulting graph, on which any number of graph partitioning schemes can be used~\cite{POF:98}. In robotic contexts, the discrete partitioning problem is often considered under communication constraints, e.g., pairwise gossip~\cite{JWD-RC-PF-FB:09w} or asynchronous one-to-base station communication~\cite{RP-PF-JWD-RC-FB:13w}. The proposed  partitioning scheme most closely mirrors~\cite{RP-PF-JWD-RC-FB:13w}; however, our approach employs additional logic to ensure the resultant coverage regions are amenable for use in decomposition-based surveillance.

Single-agent path planners for persistent surveillance most commonly operate over discrete environments, i.e., graphs~\cite{GL:92,PT-DV:01}, and solutions to classical problems, e.g., the Traveling Salesperson Problem~\cite{GG-APP:07}, may suffice. Stochasticity can be introduced using tools such as Markov chains~\cite{VS-FP-FB:11za}. Persistent surveillance strategies for non-discrete regions (in particular, open subsets of Euclidean space), are less common. Here, strategies include the \emph{a priori} construction of motion routines~\cite{JO-MGE:07}, the adaptation of static coverage strategies~\cite{DES-MS-DS:12}, the use of random fields~\cite{XL-MS:13}, and spectral decomposition~\cite{GM-AS-IM:10}. The modular framework herein does not require any particular single-agent route planner, but rather can incorporate any planner satisfying a mild assumption set (see Section~\ref{sec:routing}).

Despite the vast research on both topics individually, remarkably few papers explicitly consider the implications of combining geometric partitioning with continuous routing in the context of multi-agent persistent surveillance. 
Research that does exist is mostly preliminary, considering ideal communication and employing simplistic methods. For example, the authors of~\cite{JFA-PBS-JBS:13} employ a sweeping algorithm for partitioning and guide vehicle motion via lawn-mower pattens. The authors of~\cite{NN-IK:08} use rectangular partitions, while employing a reactive routing policy which, in ideal cases, reduces to spiral search patterns. The work in~\cite{IM-AO:07} uses slightly more sophisticated partitioning in tandem with lawn-mower motion trajectories. Also relevant is~\cite{JW-JKH:11}, where partitions are based on the statistical expectation of target presence; however, ideal communication is assumed. Other works, e.g.~\cite{BB-MV-JPH:08}, employ decomposition-based structures, but focus on task-assignment with no detailed treatment of the combined assignment/routing protocol. 

\subsection{Contributions}
\label{sec:contributions}
This work presents a decomposition-based, multi-agent coverage control framework for persistent surveillance, which requires only sporadic, unscheduled exchanges between agents and a central base station. In particular, we focus on developing a sophisticated partitioning and coordination scheme that is designed for pairing with generic single-agent trajectory planners within a modular framework. Our setup encompasses realistic constraints including restrictive communication, dynamic environments, and non-parametric event likelihoods. 

Specifically, we develop a dynamic partitioning scheme that assumes only asynchronous, \emph{one-to-base station} communication. This algorithm governs region assignments, while also introducing timing variables and manipulating the agents' high-level surveillance parameters. We prove that our partitioning strategy has properties that make it amenable for use in a decomposition-based surveillance scheme: the produced coverage regions collectively form a connected $m$-covering, local likelihood functions have disjoint support, no subregion remains uncovered indefinitely, coverage regions evolve at a time-scale that allows for appropriate agent reactions, among others. 
For static likelihoods under mild assumptions, we show that the set of coverage regions and associated generators converges to a Pareto optimal partition in finite time. 
We combine our partitioning scheme with a generic single-agent trajectory, and show that this combination guarantees collision avoidance (no sensing overlap) when the trajectory planner obeys natural restrictions. We illustrate our framework through numerical examples. 

For clarity and readability in what follows, we postpone all Theorem proofs until the Appendix.

\section{Mission Objectives and Solution Approach}
\label{sec:objectives}
A team of $m$ mobile agents, each equipped with an on-board sensor, is tasked with endlessly monitoring a non-trivial, planar, large region of interest. The primary goal of the mission is to collect sensor data about some key dynamic event or characteristic, e.g., an intruder. Data collected by the agents is periodically relayed to a central base station for analysis.
Agents must move within the region to obtain complete sensory information.  
Ideally, the agent motion/data collection strategy should be coordinated so that:  
\begin{enumerate}
\item the workload is balanced across agents,
\item no subregion goes unobserved indefinitely,
\item agents never collide (have sensor overlap),  and
\item the search is biased toward regions of greater interest.
\end{enumerate}
To achieve these goals, we adopt a decomposition-based approach in which each agent's motion is restricted to lie within a dynamically assigned \emph{coverage region}. 
\begin{figure}
\centering
\includegraphics[width = 0.9\columnwidth]{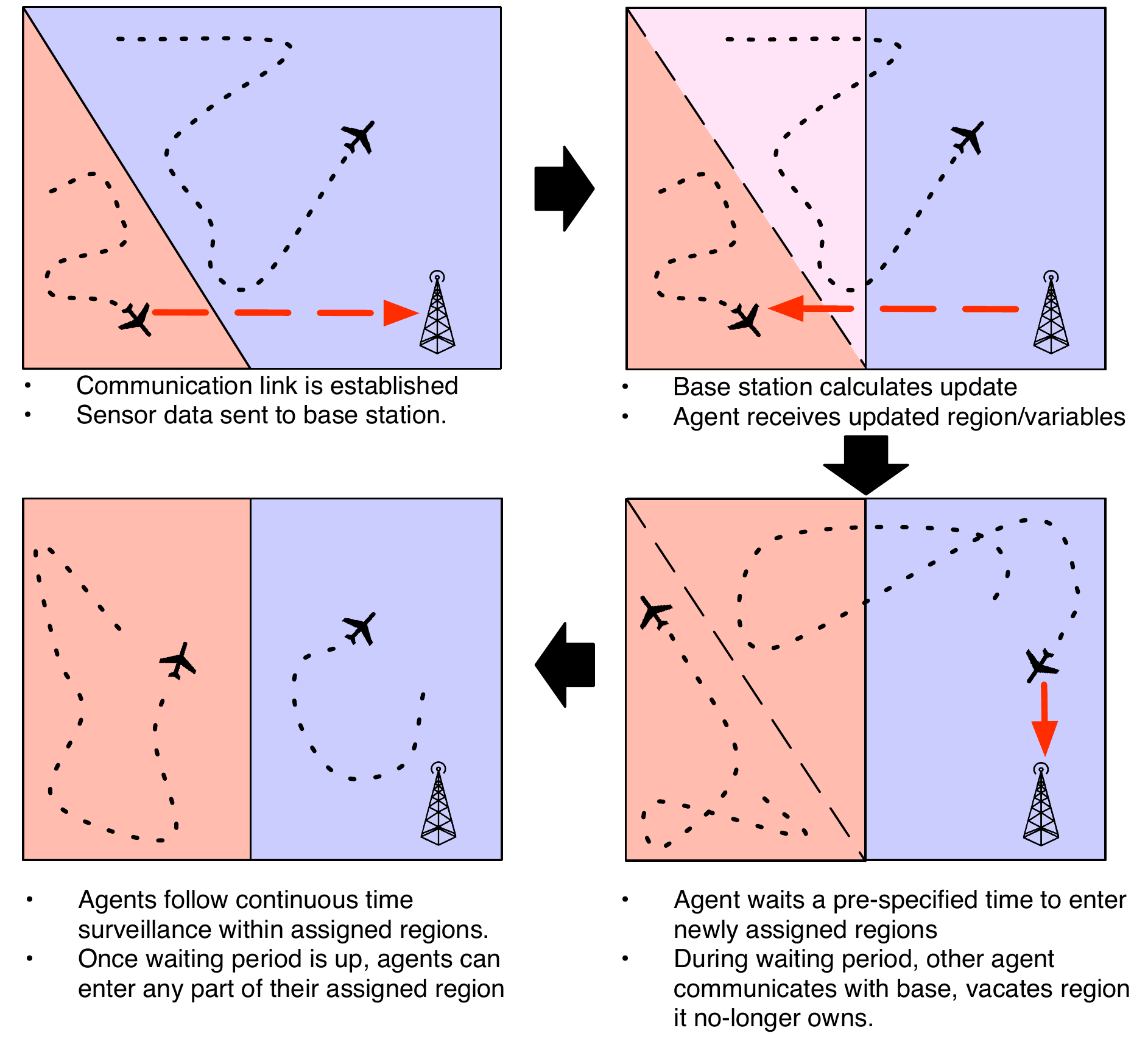}
\caption{An example illustrating the proposed strategy. The dynamic partitioning component (executed by the base station) manages coverage regions and introduces logic to prevent collisions, while the routing component (executed onboard each agent) governs agent motion.}
\label{fig:schematic}
\end{figure}

The analysis herein focuses primarily on developing an algorithmic framework for partitioning that defines, maintains, and updates the agents' coverage assignments in a manner that can be effectively paired with a typical single-agent route planner. More specifically, we seek a scheme to assign coverage responsibilities, as well as provide high-level restrictions on agent motion, so that lower-level, on-board trajectory planners can operate independently within the proposed boundaries to achieve global surveillance goals. In addition, we seek algorithms that do not require peer-to-peer communication or require the agents to be in constant contact with the base station. As such, we assume only asynchronous, \emph{one-to-base} station communication (data transfer)~\cite{RP-PF-JWD-RC-FB:13w}, i.e., agents are only assumed to sporadically communicate with a central base station, subject to an upper bound on inter-communication times, which are not specified \emph{a priori}.

%
%
%

Broadly, our proposed solution operates as follows (see Fig.~\ref{fig:schematic}). During each exchange with an agent, the base-station calculates the new coverage region solely for the communicating agent, and transmits the update. Updates are performed so as to guarantee connected coverage regions, and, in the limit for static likelihoods, distribute the coverage load among the agents. During the exchange, the base station also updates timing variables and governs the local likelihood functions maintained by each agent so that, if no agent were ever located outside of the support of its local likelihood at any time, then the agents (sensors) would never be colocated. Timing considerations also guarantee that 1) no region remains uncovered, i.e., outside of the allowable area to which the agents can travel, for an unbounded amount of time, and 2) ample time is allotted for agents to vacate regions that are re-assigned before the newly assigned agent takes over.  
Once the update is complete, the link with the base is terminated and the agent follows the trajectory found via its onboard planner. 

\section{Problem Setup}
\label{sec:definitions}
The base station, as well as each agent, has its own local processor. 
``Global'' information is stored by the base station, while each agent only stores information pertinent to itself. 
\begin{convention}
The subscripts $i, j$, or $\ell$ denote an entity relevant to agent $i,j$, or $\ell$, resp. When attached to a set, the subscripts $i,j,\ell$ indicate the component relevant to agent $i,j$ or $\ell$, resp. The superscript `$A$' indicates an entity that is stored on the agent's local processor.
\end{convention}
A summary of the variables stored by relevant entities is shown in Table~\ref{tab:storage}. We expand on these, and define other relevant mathematical constructs here.

\subsection{Agent Dynamics}
\label{sec:dynamics}

Each agent (sensor) $i$ is modeled as a point mass that moves with maximum possible speed $s_i>0$. Define $s \coloneqq \{s_i\}_{i = 1}^m$.

\subsection{Communication Protocol}
\label{comm}
There is a central base station with which the agents communicate. We assume the following:
 \begin{enumerate}
 \item each agent can identify itself to the base station and transmit sensor data,
  \item there is a lower bound $\underline{\Delta}>0$ on the time between any two successive exchanges involving the base station, and
 \item there is an upper bound $\overline{\Delta}>0$ on the time between any \emph{single agent's} successive exchanges with the base. 
 \end{enumerate}
 Note that 2) implies that no two agents can communicate with the base station simultaneously\footnote{This constraint also prevents \emph{zeno behavior} in the communication times}. These communication specifications collectively define the \emph{one-to-base station} protocol.

\subsection{Surveillance Region}
\label{sec:environment}

Consider the surveillance region as a finite grid of disjoint, non-empty, simply-connected subregions. We represent the environment as a weighted graph $G(Q)\coloneqq (Q,\mathcal{E})$, where $Q$ is the set of vertices (each representative of a unique grid element), and $\mathcal{E}$ is the edge set comprised of undirected, weighted edges $\{k_1,k_2\}$ spanning vertices representing \emph{adjacent regions}\footnote{Travel between the regions without entering other subregions is possible}. The weight assigned to $\{k_1, k_2\}$ is the distance between the appropriate regions.  We assume without loss of generality that $Q \subset \mathbb{N}$. Environmental information, e.g., terrain, is largely unknown at the mission onset; however, the location of obstacles and prohibited areas, is known \emph{a priori} and is not included in the graphical representation $G(Q)$. 

Consider $Q' \subseteq Q$. A vertex $k_1\in Q$ is \emph{adjacent} to $Q'$ if $k_1 \notin Q'$ and there exists $\{k_1, k_2\} \in \mathcal{E}$ with $k_2 \in Q'$. Define $G(Q') \coloneqq (Q',\mathcal{E}')$, where $\mathcal{E}' \coloneqq \{\{k_1, k_2\} \in \mathcal{E} | k_1, k_2 \in Q'\}$. A \emph{path} on $G(Q')$ between $k_1, k_n \in Q'$ is a sequence $(k_1, k_2, \ldots, k_n)$, where $k_1, k_2, \ldots, k_n \in Q'$ and $\{k_r,k_{r+1}\} \in \mathcal{E}'$ for $r \in \until {n-1}$. We say $Q'$ is \emph{connected} if a path exists in $G(Q')$ between any $k_1, k_2 \in Q'$.
%
%
%
%
Let $\map{d_{Q'}}{Q'\times Q'}{\mathbb{R}_{\geq 0} \cup \{\infty\}}$ be the standard distance on $G(Q')$, i.e., the length of a shortest weighted path in $G(Q')$ (if none exists, $d_{Q'}$ takes value $\infty$). Notice $d_{Q'}(k_1, k_2) \leq d_{Q}(k_1, k_2)$ for any $k_1, k_2 \in Q'$. Also let $d_{Q'}$ denote the map $\map{d_{Q'}}{Q'\times 2^{Q'}}{\real_{\geq 0} \cup \{\infty\}}$, where $d_{Q'}(k,Q'')$ is the length of a shortest weighted path in $G(Q')$ between $k$ and any vertex in $Q''$.

 \subsection{Coverage Regions}
\label{sec:partition_setup}

An \emph{$m$-covering} of $Q$ is a family $P=\{P_i\subseteq Q\}_{i = 1}^m$ satisfying
1) $\bigcup_{i=1}^m P_i = Q$, and
2) $P_i \neq \emptyset$ for all $i$. Define $\text{Cov}_m(Q)$ as the set of all possible $m$-coverings of $Q$.
An \emph{$m$-partition} of $Q$ is an $m$-covering that also satisfies
3) $P_i\bigcap P_j = \emptyset$, if $i\neq j$.
An $m$-covering or $m$-partition $P$ is $\emph{connected}$ if each $P_i$ is connected.
In what follows, the base station maintains an $m$-covering $P$ of $Q$, and surveillance responsibilities are assigned by pairing each agent $i$ with $P_i \in P$ (called agent $i$'s \emph{coverage region}). Each agent maintains a copy $P_i^A$ of $P_i$. The base station also stores a set $c \coloneqq\{c_i \in Q\}_{i = 1}^m$ ($c_i$ is the \emph{generator} of $P_i$), and each agent $i$ maintains a copy $c_i^A$ of $c_i$.

 \subsection{Identifiers, Timers, and Auxiliary Variables}
 \label{sec:timing}
 
The proposed algorithms introduce logic and timing considerations to ensure an effective overall framework. 
To each $k \in Q$, assign an \emph{identifier} $\text{ID}_k \in \until m$. Define $\mathcal{ID} \coloneqq \{\text{ID}_k\}_{k = 1}^ {|Q|}$, and let $P^{\mathcal{ID}} : = \{P^\mathcal{ID}_{i}\}_{i = 1}^m$, where $P^\mathcal{ID}_{i}: = \{k \in Q\mid \text{ID}_k = i\}$. Notice $P^{\mathcal{ID}}$ is an $m$-partition of $Q$. 
For each agent $i$, define a timer $T_i$ having dynamics 
$\dot{T}_i = -1$ if $T_i \neq 0$, and $\dot{T}_i = 0$ otherwise. The set $T\coloneqq\{T_i\}_{i=1}^{m}$ is maintained by the base station. Similarly, each agent $i$ maintains a local timing variable $\tau^A_i$. Even though $\tau^A_i$ plays a similar role to the timer $T_i$, note that $\tau^A_i$ is constant unless explicitly changed by the algorithm, while $T_i$ has autonomous dynamics. 
Next, the base station maintains a set $\omega \coloneqq \{\omega_i\}_{i = 1}^m,$ where $\omega_i$ is the time of agent $i$'s last exchange with the base station. Each agent maintains a copy $\omega_i^A$ of $\omega_i$. Finally, each agent stores locally a subset $\supscr{P_i}{$A$,pd} \subseteq P_i^A$ which, loosely, collects vertices that have recently been added to $P_i^A$. 

\subsection{Likelihood Functions}
\label{sec:prior}
A real-time estimate of the likelihood of events of interest occurring within any subset of the surveillance region is maintained by the base station in the form of a time-varying probability mass function\footnote{For any $t$, we have $\sum_{k \in Q} \Phi(k,t)$ = 1} $\map{\Phi}{Q \times \real_{\geq 0}}{\real_{\geq 0}}$.
Since it is difficult to precisely define load-balancing requirements over arbitrarily time-varying environments, our analysis focuses on achieving optimal performance in the case of a static environment, i.e., when $\Phi$ is time-invariant. We then show through simulation that our proposed strategy still can be effectively employed in the case of a \emph{quasi-static} environment, i.e., when $\Phi(k,\cdot)$ is piecewise constant for any $k$. This situation arises in many common scenarios, e.g., when estimates of $\Phi$ are only updated during exchanges between the base station and the mobile agents, which occur at discrete time-points.

 Define each agent's \emph{local likelihood} $\map{\Phi^A_i}{Q \times \real_{\geq 0} }{\real_{\geq 0}}$ as the function that, loosely, represents the agent's local belief regarding events of interest occurring in $Q$. The partitioning algorithm will manipulate the support of each $\Phi^A_i$ to ensure effective surveillance while also inherently preventing agent collisions (sensor redundancy). Specifically, define
\begin{multline}
\Phi^A_{i}(k,t) = \begin{cases} \Phi(k,t),& \begin{array}{c}\text{ if }k \in P_i^A\text{ and}\\ \left(\cramped{t-\omega^A_i\geq\tau^A_{i}\text{ or } k \notin \supscr{P_i}{$A$,pd}}\right),\end{array}  \\
\\ 0,&\text{otherwise}. \end{cases}
\label{eqn:phi_temp}
\end{multline}
In general, each function $\Phi^A_i$ will be different\footnote{Note that $\Phi^A_i$ need not be normalized and thus may not be a time-varying probability mass function in a strict sense} from $\Phi$. 
 \begin{remark}[Global Data]
If global knowledge of $\Phi$ is not available instantaneously to agent $i$, $\Phi^A_i$ can alternatively be defined by replacing $\Phi(k,t)$ in~\eqref{eqn:phi_temp} by $\Phi(k, \omega_i^A)$. All subsequent theorems hold under this alternative definition.
\label{rem:global}
\end{remark}
\begin{table}
\centering
 \caption{Storage Summary}
 \begin{tabular}{cl}
 \toprule
 \multicolumn{2}{c}{Stored by Base Station} \\
 \toprule
 Variable & Description \\
 \midrule
 $G(Q)$&Graphical representation of the environment\\
 $P$ & $m$-covering of $Q$ $\left(P\in \text{Cov}_m(Q)\right)$\\
 $c$ & Set of generators $\left(c \in Q^m\right)$ \\
 $\mathcal{ID}$ & Set of identifiers $\left(\mathcal{ID} \in {\until m}^{|Q|}\right)$\\ 
 $T$ & Set of Timers $\left(T(t) \in \real_{\geq 0}^{m}\right)$ \\
 $\omega$ & Set of most recent communication times $\left(\omega \in \mathbb{R}^m_{\geq 0}\right)$ \\
 $\Phi$ & Global likelihood function $\left(\map{\Phi}{Q\times \real_{\geq 0}}{\real_{\geq 0}}\right)$ \\
 \toprule
 \multicolumn{2}{c}{Stored by Agent $i$} \\ 
 \toprule
  Variable & Description \\
 \midrule
  $G(Q)$&Graphical representation of the environment\\
 $P_i^A$ & Coverage region $\left(P_i^A \subset Q\right)$\\
 $c_i^A$ & Coverage region generator $\left(c_i^A \in Q\right)$\\
 $\supscr{P_i}{$A$,pd}$ & Set of ``recently added'' vertices $\left(\supscr{P_i}{$A$,pd} \subseteq P_i^A\right)$ \\
 $\tau^A_i$ & Local timing parameter $\left(\tau^A_i \in \real\right)$ \\
 $\omega^A_i$ & Agent $i$'s most recent communication time $\left(\omega^A_i \in \real_{\geq 0}\right)$ \\
 $\Phi^A_i$ & Local likelihood function $\left(\map{\Phi^A_i}{Q\times \real_{\geq 0}}{\real_{\geq 0}}\right)$\\\bottomrule
 \end{tabular} 
 \label{tab:storage}
 \end{table}
%
%

\begin{remark}[Data Storage]
The cost of storing a graph as an adjacency list is $O(|Q|+|\mathcal{E}|)$. The generator set $c$,  each element of $P$, and the identifier set $\mathcal{ID}$ are stored as a integral vectors. The timer set $T$ and the set $\omega$ are are stored as real vectors. Typically, only the values of $\Phi$ at vertices in $Q$ are needed, so $\Phi$ is stored as a time-varying real vector. Thus, the cost of storage at the base station is $O(m|Q|+|\mathcal{E}|)$. Similarly, each agent's local storage cost is $O(|Q|+|\mathcal{E}|)$.
\label{rem:data_storage}
\end{remark}


 \section{Dynamic Coverage Update Scheme}
 \label{sec:partitioning}
%
We adopt following convention for the remaining analysis.
 \begin{convention}
Suppose that:
 \begin{enumerate} 
 \item $\min \emptyset \coloneqq \max \emptyset \coloneqq 0$, and
 \item when referring to a specific time instant, e.g., the instant when an update is executed, the notation $e^+$ and $e^-$ refers to the value of the entity $e$ immediately before and after the instant in question, resp.
 \end{enumerate}
 \label{con:timed}
 \end{convention}

 \subsection{Additive Subset}
\label{sec:additive}
We start with a definition.
\begin{definition}[Additive Subset]
Given $k \in P_i^{\mathcal{ID}}$, the \emph{additive subset} $\supscr{P_i}{add}(k) \subseteq Q$ is the largest connected subset satisfying:
\begin{enumerate}
\item $P_i^{\mathcal{ID}} \subseteqq \supscr{P_i}{add}(k)$, and 
\item for any $h\in \supscr{P_i}{add}(k) \cap P_j$, where $j \neq i$:
\begin{enumerate}
\item $T_j = 0$, and 
\item $\frac{1}{s_i}d_{\supscr{P_i}{add}(k)}(h, k) < \frac{1}{s_j}d_{P_j}(h, c_j)$.
\end{enumerate}
\end{enumerate}
\label{def:P_plus}
\end{definition}
The following characterizes well-posedness of Definition~\ref{def:P_plus}.
\begin{proposition}[Well-Posedness]
If $P_i^{\mathcal{ID}}$ is connected and disjoint from $\bigcup_{j \neq i} P_j$, then $\supscr{P_i}{add}(k)$ exists and is unique for any $k \in P_i^{\mathcal{ID}}$.
\label{prop:well-posed}
\end{proposition}
\begin{proof}
With the specified conditions, $P_i^{\mathcal{ID}}$ is connected and satisfies conditions $1$-$3$ in Definition~\ref{def:P_plus}.  Thus, $\supscr{P_i}{add}(k)$ is the unique, maximally connected superset of $P_i^{\mathcal{ID}}$ satisfying the same conditions.
\end{proof}
Under the conditions of Proposition~\ref{prop:well-posed}, if $h \in\supscr{P_i}{add}(k)$, then $\max\{T_j|j \neq i, h \in P_j\} = 0$ and there is a path from $k$ to $h$ in $G(\supscr{P_i}{add}(k))$ that is shorter than the optimal path spanning $c_j$ and $h$ within $G(P_j)$, for any $j\neq i$ with $h \in P_j$. Also note that, by definition, additive subsets are connected.

\subsection{Primary Update Algorithm}
\label{sec:primary}
Define $\map{\mathcal{H}}{ Q^m\times \text{Cov}_m(Q)\times \real_{\geq 0}}{\real_{\geq 0} \cup \{\infty\}}$ by 
\begin{equation*}
\mathcal{H}(c,P,t) = \sum_{k \in Q}\min\left\{\frac{1}{s_i}d_{P_i}(k,c_i)|k \in P_i\right\}\Phi(k,t),
\end{equation*}
where $\Phi$ is the global likelihood. If (i) each agent is solely responsible for events within its own coverage region and (ii) events occur proportionally to $\Phi$, then $\mathcal{H}(c,P,t)$ is understood as the expected time required for an appropriate agent to reach a randomly occurring event from its region generator at time $t$. 
Algorithm~\ref{alg:partitioning} defines the operations performed by the base station when agent $i$ communicates at time $t_0$. 
\begin{algorithm}
\DontPrintSemicolon
\KwData{$t_0$, $P$, $c$, $\Phi$, $\omega$, $\overline{\Delta}$, $\subscr{\Delta}{H}$, $T$, $\mathcal{ID}$, $s$}
\KwResult{$P$, $c$, $P_i^A$, $c_i^A$, $\supscr{P_i}{$A$,pd}$, $T$,   $\tau^A_i$,  $\Phi^A_i$, $\omega$, $\omega_i^A$,  $\mathcal{ID}$ }
\Begin{
\nl Initialize $P^* = \supscr{P}{test} = P$ , $c^* = \supscr{c}{test} =  c$\;
\nl Set $P^*_i = \supscr{P}{test}_i = P^{\mathcal{ID}}_i$\;
\eIf{$T_i > 0$ and $P_i^* = P_i$}{
	\nl Set $\tau_i^A = \tau_i^A - t_0+\omega_i$ and $\omega_i^A = \omega_i = t_0$\;
}{
\For{$k \in P^{\mathcal{ID}}_i$}
{\nl Set $\supscr{P_i}{test} = \supscr{P_i}{add}(k)$ and $\supscr{c_i}{test} = k$\;
\If{$\mathcal{H}(\supscr{c}{test},\supscr{P}{test}, t) < \mathcal{H}(c^*,P^*, t)$}{\nl Set $P^* = \supscr{P}{test}$, $c^* = \supscr{c}{test}$}}
\nl Set $\supscr{P_i}{$A$,pd} =  P_i^* \backslash P^{\mathcal{ID}}_i$ \;
\nl Call Alg.~\ref{alg:timer} and obtain output $\Phi_i^A, \omega, T, \tau_i^A$\;
\nl Set $P_i = P_i^A = P_i^*$, $c_i =  c_i^A  = c_i^*$, $\omega_i^A = \omega_i$\;
 \nl\lFor {$k \in P_i$}{Set $\text{ID}_k = i$}
 }
 \nl \Return $P$, $c$, $P_i^A$, $c_i^A$, $\supscr{P_i}{$A$,pd}$, $T$, $\tau_i^A$, $\Phi_i^A$, $\omega$, $\omega_i^A$, $\mathcal{ID}$}
 
\caption{Timed One-To-Base Station Update}
\label{alg:partitioning}
\end{algorithm}
Here, the input $\subscr{\Delta}{H}>0$ is a constant mission-specific parameter.

\begin{algorithm}
\DontPrintSemicolon
\KwData{$t_0$, $P$, $P^*$, $c^*$, $\supscr{P_i}{$A$,pd}$, $\Phi$, $\omega$, $\overline{\Delta}$, $\subscr{\Delta}{H}$, $T$, $s$}
\KwResult{$\Phi_i^A, \omega, T, \tau_i^A$}
\Begin{
\nl $\supscr{\Delta_i}{Bf} \coloneqq \max\left\{\frac{1}{s_i}d_{P_i}(k, P_i^* \backslash \supscr{P_i}{$A$,pd})|k \in P_i\backslash P_i^*\right\}$\;
\For{Each $j \neq i$ satisfying $P_j \cap P_i^*\neq \emptyset$}{
\nl $\supscr{\Delta_j}{Bf} \coloneqq \max\left\{\frac{1}{s_j}d_{P_j}(k, P_j\backslash P_i^*)|k \in P_j \cap P_i^*\right\}$\;
\nl Set $T_j = \omega_j + \overline{\Delta} - t_0$\;
}
\nl Find $\cramped{\supscr{\Delta_{\text{max}}}{Bf} = \underset{j \neq i, P_j \cap P_i^* \neq \emptyset}{\max}\{\omega_j+\overline{\Delta} + \supscr{\Delta_j}{Bf}- t_0\}}$\;
\nl Redefine $\supscr{\Delta_{\text{max}}}{Bf} = \max\{\supscr{\Delta_{\text{max}}}{Bf}, \supscr{\Delta_i}{Bf}\}$\;
\nl Set $T_i = \supscr{\Delta_{\text{max}}}{Bf}+\subscr{\Delta}{H}$, $\tau_i^A = \supscr{\Delta_{\text{max}}}{Bf}$, $\omega_i = t_0$\;
\nl Construct $\Phi_i^A$, according to~\eqref{eqn:phi_temp} with updated variables\;
\nl \Return $\Phi_i^A, \omega, T, \tau_i^A$
 }
\caption{Timer Update}
\label{alg:timer}
\end{algorithm}
Consider the following initialization assumptions. 
\begin{assn}[Initialization]
The following properties are satisfied when $t = 0$:
\begin{enumerate}
\item $P$ is a connected $m$-partition of $Q$, 
\item $P = \supscr{P}{A} = P^{\mathcal{ID}} $, and
\item for all $i \in \until m$,
\begin{enumerate}
\item $c_i = c_i^A\in P_i^A$, 
\item $\supscr{P_i}{$A$,pd} = \emptyset$, 
\item $T_i = \omega_i  = \omega_i^A = 0$,
\item $\tau_i^A = -\subscr{\Delta}{H}$, and
\item $\Phi_i^A(\cdot, 0)$ is defined according to Eq.~\ref{eqn:phi_temp}.
\end{enumerate}
\end{enumerate}
\label{assn:initialization}
\end{assn}
Notice 1) and 3a) together imply that $c_i \neq c_j$ for any $j \neq i$. Our first result guarantees that the use of Algorithm~\ref{alg:partitioning} for coverage region updates is well-posed.
\begin{theorem}[Well-Posedness]
Under Assumption~\ref{assn:initialization}, an update scheme in which, upon each exchange, the base station and the communicating agent update their respective variables via Algorithm~\ref{alg:partitioning} is well-posed. That is, the operations required by Algorithm~\ref{alg:partitioning} are well-posed at the time of execution.
\label{thm:well-posedness}
\end{theorem}

Performing of updates with Algorithm~\ref{alg:partitioning} does not guarantee that individual coverage regions (elements of $P$), are disjoint from one another. 
It does, however, guarantee that the $m$-covering $P$, the local coverage region set $P^A\coloneqq\{P_i^A\}_{i = 1}^m$, and the local likelihoods $\{\Phi^A_i\}_{i = 1}^m$ retain properties that are consistent with the goals of the decomposition-based coverage framework. Namely, the coverings $P$ and $P^A$ maintain connectivity, and each function $\Phi^A_i$ has support that is disjoint from that of all other local likelihoods, yet still evolves to provide reasonable global coverage. The manipulations in Algorithm~\ref{alg:timer} also ensure that agents are able to ``safely'' vacate areas that are re-assigned before newly assigned agents enter. We detail these and other properties for the remainder of this section.

\subsection{Set Properties}
\label{sec:set}

The next result formalizes key set properties.
\begin{theorem}[Set Properties]
Suppose Assumption~\ref{assn:initialization} holds, and that, upon each exchange, the base station and the communicating agent update their respective variables via Algorithm~\ref{alg:partitioning}. Then, the following\footnote{$\text{supp}(f)$ denotes the support of a function $f$.} hold at any time $t \geq 0$:
\begin{enumerate}
\item $P^{\mathcal{ID}}$ is a connected $m$-partition of $Q$,
\item $P$ is a connected $m$-covering of $Q$,
\item $c_i \in P_i$ and $c_i \neq c_j$ for any $i \neq j$,
\item $\text{supp}(\Phi_i^A(\cdot, t)) \subseteq P_i$ for any $i$, and 
\item $\bigcap_{i=1}^m \text{supp}(\Phi_i^A(\cdot,t)) = \emptyset$
\end{enumerate}
\label{thm:set_properties}
\end{theorem}
Whenever additions are made to an agent's coverage region during a call to Algorithm~\ref{alg:partitioning}, the
newly added vertices are not immediately included in the instantaneous support\footnote{The \emph{instantaneous support} of $\Phi_i^A$ at time $t$ is defined as $\text{supp}(\Phi_i^A(\cdot, t))$.} of the agent's local likelihood. As such, if each agent's movement is restricted to lie within the aforementioned instantaneous support, then exploration of newly added regions is temporarily prohibited following an update. This delay allows other agents to vacate before the newly assigned agent enters. 
Conversely, when regions are removed from an agent's coverage region, Algorithm~\ref{alg:partitioning} guarantees a ``safe'' path, i.e., a path with no collision risk, exists and persists long enough for the agent to vacate.
Let $\overline{d}\coloneqq \max_i\frac{1}{s_i}\sum_{\{k_1,k_2\} \in \mathcal{E}} d_{Q}(k_1,k_2)$, and define the \emph{prohibited region} of agent $i$ at time $t$, $\text{Proh}_i(t)$, as the subset consisting of any newly added vertices that do not yet belong to  $\text{supp}(\Phi_i^A(\cdot, t))$, i.e., $\text{Proh}_i(t)\coloneqq \{k \in \supscr{P_i}{$A$}|, t-\omega^A_i<\tau_i^A \text{ and } k \in \supscr{P_i}{$A$,pd}\}$, we formalize this discussion as follows.
\begin{theorem}[Coverage Quality]
Suppose Assumption~\ref{assn:initialization} holds, and that, upon each exchange, the base station and the communicating agent updates their respective variables via Algorithm~\ref{alg:partitioning}. Then, for any $k \in Q$ and any $t \geq 0$:
 \begin{enumerate}
 \item $k$ belongs to at least one agent's coverage region $P_i$, 
 \item if $k \in \text{Proh}_i(t)$ for some $i$, then there exists $t_0$ satisfying $t< t_0 <t+\overline{\Delta}+\overline{d}$ such that, for all $\bar{t} \in [t_0,t_0+\subscr{\Delta}{H}]$, the vertex $k$ belongs to the set $P_i\backslash \text{Proh}_i(\bar{t})$, and 
\item if $k$ is removed from $P_i$ at time $t$, then, for all \[\bar{t} \in\left.\left(t, t+\frac{1}{s_i}d_{P_i^-}\left(k, P^{\mathcal{ID},-}_i\right)\right.\right],\] we have
\begin{enumerate}
\item $P^{\mathcal{ID},-}_i\subseteq P_i$, and 
\item there exists a length-minimizing path on $G(P_i^-)$ from  $k$ into $P^{\mathcal{ID},-}_i$, and all of the vertices along any such path (except the terminal vertex) belong to the set  $\text{Proh}_{\text{ID}_k^+}(\bar{t})\backslash \bigcup_{j \neq \text{ID}_k^+} P_j$. 
\end{enumerate}
\end{enumerate}
\label{thm:characteristics}
\end{theorem} 

Theorems~\ref{thm:set_properties} and~\ref{thm:characteristics} allow Algorithm~\ref{alg:partitioning} to operate within a decomposition-based framework to provide reasonable coverage with inherent collision avoidance.
Indeed, when each agent moves within its coverage region and avoids its prohibited region, the theorems imply that each agent 1) can visit its entire coverage region (connectedness), 2) allows adequate time for other agents to vacate newly assigned regions before entering, and 3) has a ``safe'' route into the remaining coverage region if its current location is removed during an update.  
\begin{remark}[Local Variables]
Theorems~\ref{thm:set_properties} and~\ref{thm:characteristics} also hold if we replace $P$ with $P^A$ and $c$ with $c^A$ in the theorem statement. 
\label{rem:local}
\end{remark}
\begin{remark}[Bounds]
Theorem~\ref{thm:characteristics} also holds when $\overline{d}$ is redefined as any other upper bound on the subgraph distance between two arbitrarily chosen vertices of an arbitrarily chosen connected subgraph of $G$.
\label{rem:bounds}
\end{remark}
\subsection{Convergence Properties}
\label{sec:convergence}
Our proposed strategy differs from that of~\cite{RP-PF-JWD-RC-FB:13w} mainly due to the presence of logic (e.g., timer manipulations) to ensure effective pairing with single-agent trajectory planners. Note also that $\mathcal{H}$ differs from typical partitioning costs, since it uses subgraph, rather than global graph, distances. As such, convergence properties of the algorithms herein do not follow readily from existing results. Consider the following definition. 
 \begin{definition}[Pareto Optimality]
The pair $(c,P)$ is Pareto optimal at time $t$ if (i) $\mathcal{H}(c, P, t) \leq \mathcal{H}(\bar{c}, P, t)$ for any $\bar{c} \in Q^m$, and (ii) $\mathcal{H}(c, P, t) \leq \mathcal{H}(c, \bar{P}, t)$ for any $\bar{P} \in \text{Cov}_m(Q)$.
 \label{def:pareto}
 \end{definition}
 The following result characterizes the dynamic evolution of coverage regions with respect to Pareto optimality.

\begin{theorem}[Convergence]
Suppose Assumption~\ref{assn:initialization} holds and that, upon each exchange, the base station and the communicating agent updates their respective variables via Algorithm~\ref{alg:partitioning}.
If $\Phi$ is static, i.e., $\Phi(\cdot,t_1) = \Phi(\cdot, t_2)$ for all $t_1,t_2$, then the $m$-covering $P$  and the generators $c$ converge in finite time to an $m$-partition $P^*$ of $Q$ and a set $c^*$, resp. The pair $(c^*, P^*)$ is Pareto optimal at any time following convergence.
\label{thm:convergence}
\end{theorem}

Whenever the event likelihood is static (and Assumption~\ref{assn:initialization} holds), Algorithm~\ref{alg:partitioning} causes coverage regions and generators to collectively converge in finite time to a Pareto optimal configuration. That is, the agent's limiting coverage assignments are ``optimal'' in that they balance the coverage load in a manner that directly considers the likelihood $\Phi$. Further, the entire operation only relies on sporadic and unplanned information exchange between agents and the base station. 

\begin{remark}[Voronoi Partitions]
It can be shown that Pareto optimality of $(c^*, P^*)$ in Theorem~\ref{thm:convergence} implies that, following convergence, $P^*$ is a multiplicatively weighted Voronoi partition (generated by $c^*$, weighted by $s$, subject to density $\Phi(\cdot, t)$) by standard definitions (e.g.,~\cite{RP-PF-JWD-RC-FB:13w}). If the centroid set of each $P_i$ is defined as $\arg\min_{h \in P_i} \sum_{k\in P_i} d_{P_i}(k,h)\Phi(k,t)$, then $P^*$ is also centroidal. 
\label{rem:centroidal}
\end{remark}

\section{Decomposition-Based Surveillance.}
\label{sec:routing}
This section pairs the proposed partitioning framework with a generic, single-vehicle trajectory planner, forming the complete, decomposition-based coverage control framework.

\subsection{Complete Routing Algorithm} 
\label{sec:complete}

Theorems~\ref{thm:set_properties} and~\ref{thm:characteristics} provide a number of useful insights when using the dynamic partitioning updates of the previous section. Namely, Theorem~\ref{thm:set_properties} ensures that (i) the instantaneous support of each $\Phi_i^A$ lies entirely within the coverage region $P_i^A$, and (ii) the support of two distinct agents' local likelihoods do not intersect. Theorem~\ref{thm:characteristics} states that (i) length of any interval on which a given vertex is uncovered, i.e., belongs solely to agent prohibited regions, cannot exceed a finite upper bound, and (ii) the parameter $\subscr{\Delta}{H}$ is a lower bound on the length of time that a recently uncovered vertex must remain covered before it can become uncovered again. Since $G(Q)$ is a discrete representation of the surveillance region, these facts together suggest that an intelligent routing scheme that carefully restricts agent motion according to the instantaneous support of the local likelihood functions could achieve adequate coverage while also accomplishing the ancillary goal of collision avoidance. This motivates the following assumption.
\begin{assn}[Agent Motion]
Each agent $i$ has knowledge of its position at any time $t$, and its on-board trajectory planner operates under the following guidelines:
\begin{enumerate}
\item generated trajectories obey agent motion constraints,
\item trajectories are constructed incrementally and can be altered in real-time, and
\item the agent is never directed to leave regions associated with $P_i$ or enter regions associated with $\text{Proh}_i(t)$.
\end{enumerate}
Each agent precisely traverses generated trajectories.
\label{assn:motion}
\end{assn}

Algorithm~\ref{alg:complete} presents the local protocol for Agent $i$.
\begin{algorithm}
\DontPrintSemicolon
\KwData{$G(Q)$, $\Phi_i^A$, $P_i^A$, $c_i^A$, $\supscr{P_i}{$A$,pd}$, $\tau_i^A$, $\omega_i^A$}
\Begin{
\While{True}{
\nl Increment trajectory via on-board planner\;
\nl Follow trajectory\;
\If{Communication with the base station}{
\nl Set $\supscr{P_i}{test} = P_i^A$\;
\nl Obtain updated variables from base station\;
\If{Location lies within the set $P_i^A\backslash\supscr{P_i}{test}$}{
\nl Find a minimum-length path in $G(\supscr{P_i}{test})$ from the currently occupied node into $P_i^A$\;
\While{Agent $i$ is outside $P_i^A$}{
\nl Follow the aforementioned route
}
}
}
}
}
\caption{Motion Protocol for Agent $i$}
\label{alg:complete}
\end{algorithm}

\subsection{Collision Avoidance}
\label{sec:collide}

Although Assumption~\ref{assn:motion} locally prevents agents from leaving assigned coverage regions or entering prohibited regions, dynamic coverage updates can still result in undesirable agent configurations.
Indeed, an agent can be placed outside of its own coverage region if the vertex corresponding to its location is abruptly removed during an update. If this happens, Algorithm~\ref{alg:complete} constructs a route from the agent's location back into a region where there is no collision risk. With mild assumptions, Theorem~\ref{thm:characteristics} guarantees that this construction 1) is well-defined, and 2) does not present the agent with a transient collision risk. We formalize this result here.
\begin{theorem}[Collision Avoidance]
Suppose Assumptions~\ref{assn:initialization} and~\ref{assn:motion} hold, and that each agent's initial position lies within its initial coverage region $P_i$. Suppose further that each agent's motion is locally governed according to Algorithm~\ref{alg:complete}, where the update in line $4$ is calculated by the base station via Algorithm~\ref{alg:partitioning}. If the weight assigned to each edge in $G(Q)$ is an upper bound on the distance between the associated regions, then no two agents will ever collide.
\label{thm:collision}
\end{theorem}

By running the schemes of Section~\ref{sec:partitioning} in conjunction with a motion planning scheme obeying Assumption~\ref{assn:motion}, we obtain a complete strategy that 1) only requires weak communication assumptions, 2) provides dynamic load-balancing, and 3) has inherent collision avoidance/efficiency properties.

\section{Numerical Examples}
\label{sec:simulations}

This section presents numerical examples to illustrate the functionality of the decomposition-based routing scheme. In all examples, high-level coverage assignment updates are performed by the base station via Algorithm~\ref{alg:partitioning} during each exchange with an agent, while each agent's local processor runs the motion protocol in Algorithm~\ref{alg:complete}. For incremental trajectory construction (Algorithm~\ref{alg:complete}, line $1$), we implement a modified version of the \emph{Spectral Multiscale Coverage} (SMC) scheme in~\cite{GM-IM:09}, which creates agent trajectories that mimic ergodic dynamics while also locally constraining agent motion to lie within the appropriate sets. This planner satisfies Assumption~\ref{assn:motion}.
Initial region generators were selected randomly (enforcing the non-coincidence constraint), and each agent was initially placed at its region generator. The initial covering $P$ was created from these generators by calculating a weighted Voronoi partition. The remaining parameters were chosen according to Assumption~\ref{assn:initialization}. During the simulation, randomly chosen agents (chosen via random number generator) sporadically communicated with the base station to receive coverage assignment updates. Communication times were randomly chosen, subject to a maximum inter-communication time $\overline{\Delta}$. 

\subsection{Static Likelihood}

The first example is a $4$ agent mission, which is executed over a 100 x 100 surveillance region subject to a static, Gaussian likelihood centered at the bottom left corner. For discretization, the region is divided into $400$, $5$ x $5$ subregions, and regions are considered adjacent if they share a horizontal or vertical edge. Each agent has a maximum speed of $1$ unit distance per unit time, and the maximum inter-communication time is $\overline{\Delta} = 10$ time units. Figure \ref{fig:Base Simulation} shows the evolution of the coverage regions at various time points for an example simulation run. Note that Figure~\ref{fig:Base Simulation} only shows each agent $i$'s \emph{active} coverage region, i.e., the subset of $P_i$ that does not intersect its prohibited region $\text{Proh}_i(t)$. The family of active coverage regions does not generally form an $m$-covering of $Q$; however, elements of this family are connected and never intersect as a result of inherent collision avoidance properties.
\begin{figure} [h]
	\centering
	\includegraphics[width = 0.85\columnwidth]{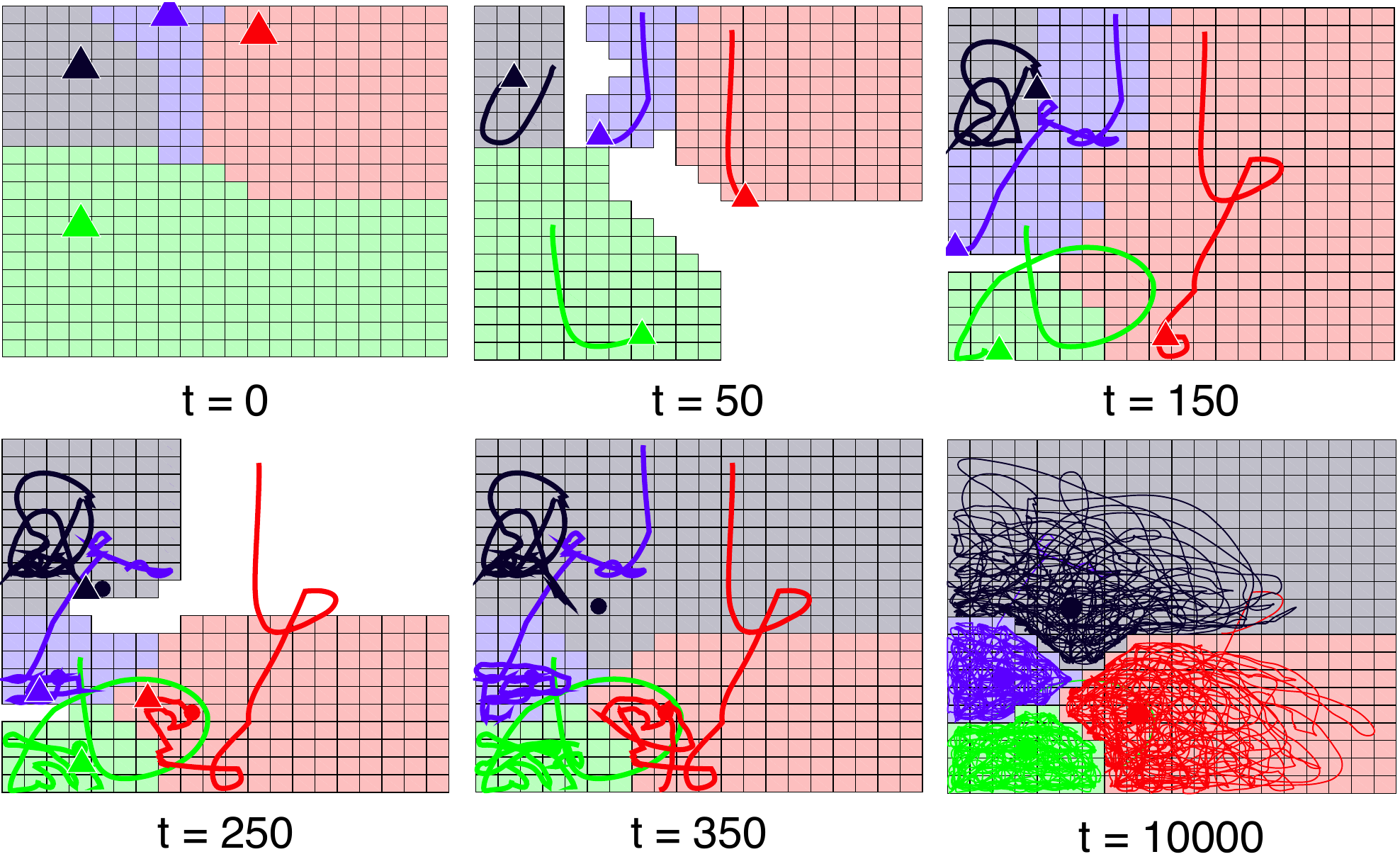}
	\caption{Snapshots of a $4$ agent surveillance mission, assuming a static Gaussian likelihood. Each agent's position, past trajectory, and \emph{active coverage region} are indicated by the colored triangle, line, and squares, resp..}
	\label{fig:Base Simulation}
\end{figure}

The left plot in Figure~\ref{fig:TimeCost} depicts the maximum amount of time that any individual subregion went uncovered, i.e. the subregion did not belong to any agent's active covering $P_i \backslash \text{Proh}_i(t)$, during each of $50$ simulation runs.  
\begin{figure} [h]
\centering
\subfigure{\includegraphics[width = 0.49\columnwidth]{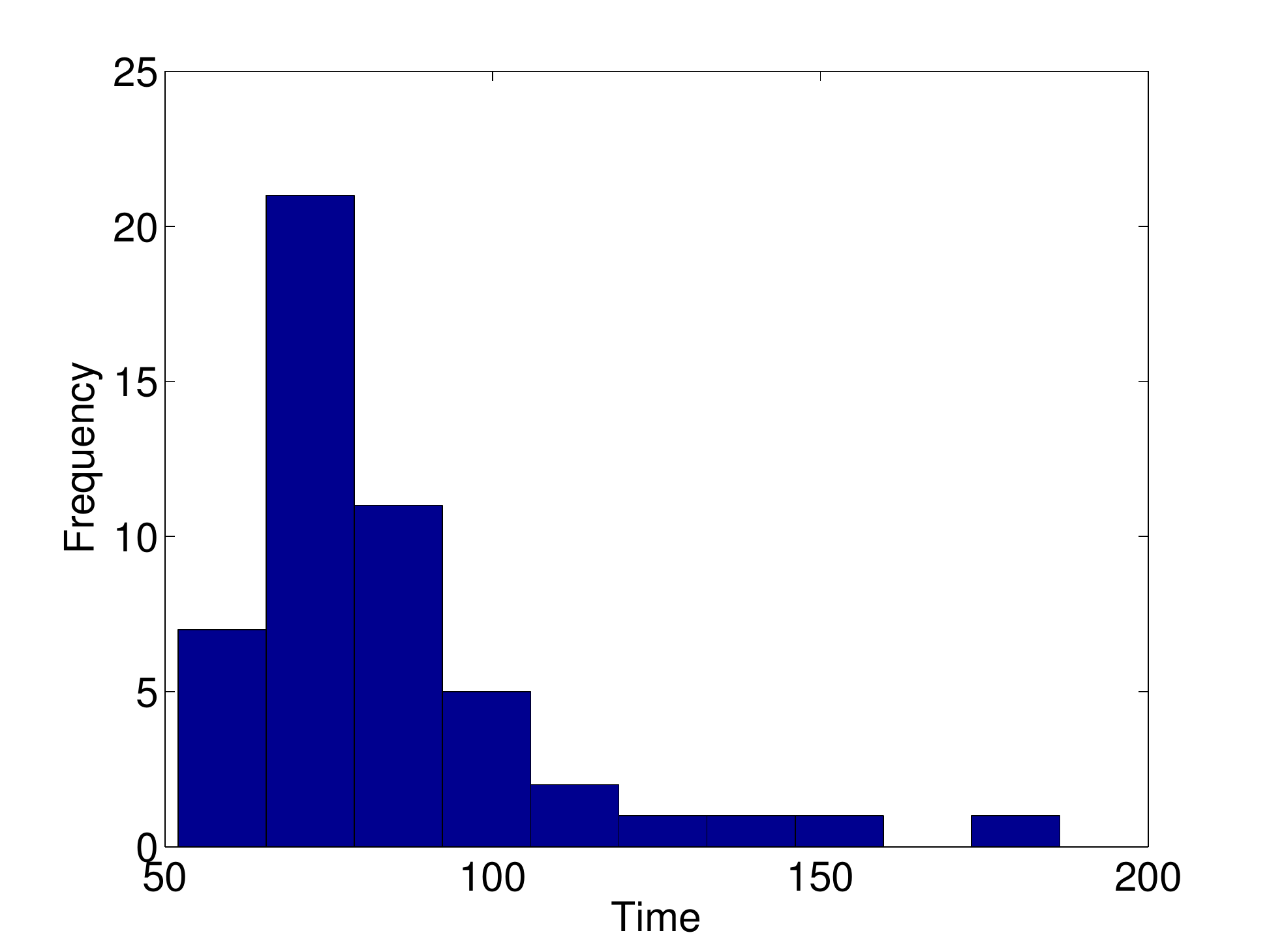}}
\subfigure{\includegraphics[width = 0.49\columnwidth]{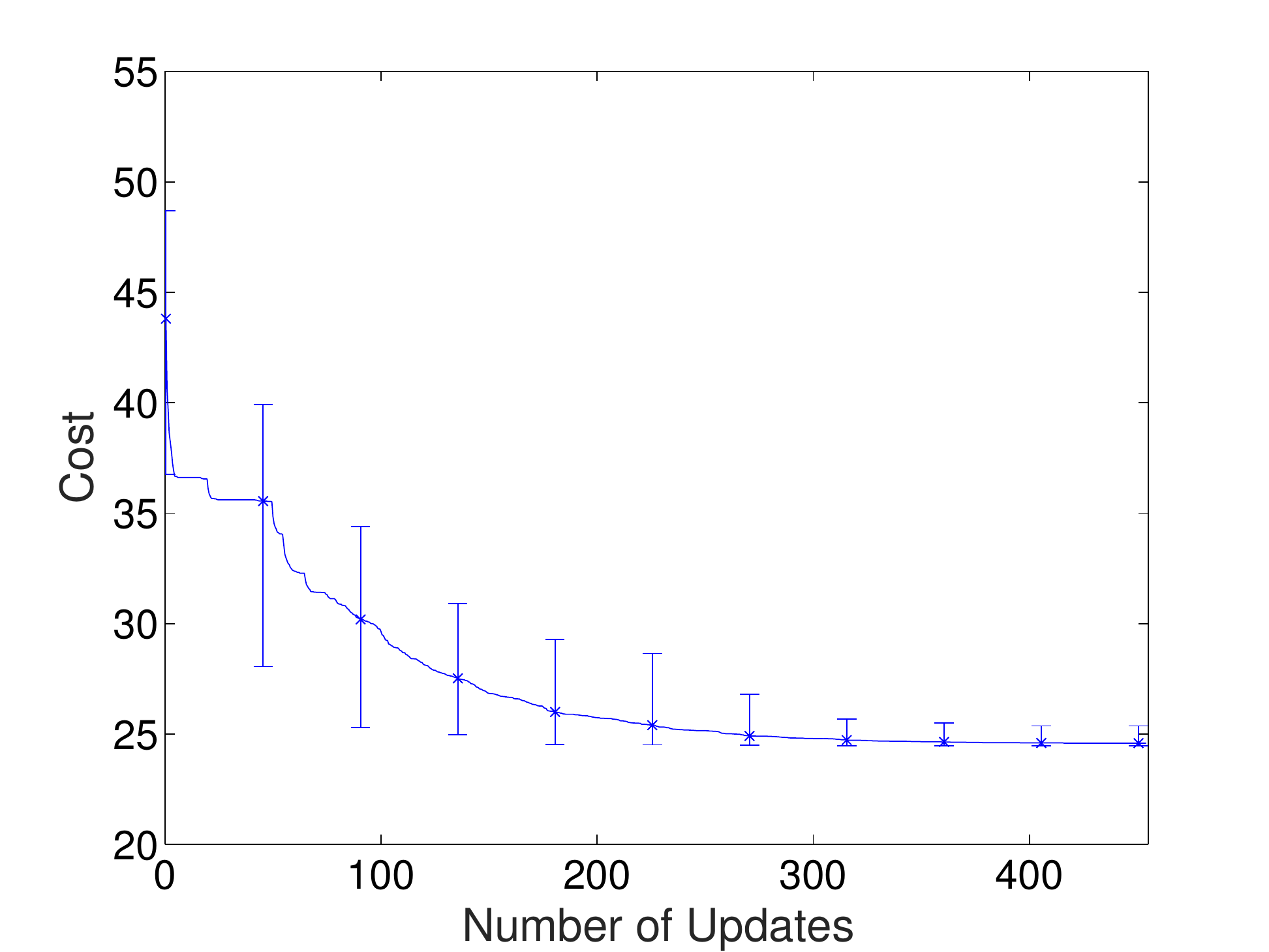}}
\caption{The maximum amount of time that any subregion went uncovered in each of $50$ simulation runs (left), and the value of the cost $H$ as a function of time, averaged over the same $50$ runs (right).}
\label{fig:TimeCost}
\end{figure}Here, the maximum amount of time that any region went uncovered was $186$ units, though most trials had maximums of less than $75$ units. This is well-below the loose bound $\overline{\Delta}+\overline{d} = 770$ predicted by Theorem~\ref{thm:characteristics} (see Remark~\ref{rem:bounds}).
The right plot in Figure~\ref{fig:TimeCost} shows the mean values of the cost function $\mathcal{H}$ as a function of time, calculated over the same $50$ simulations runs. Here, error bars represent the range of cost values achieved at select time points. The variance between runs is due to the stochastic nature of the data-exchange patterns between the agents and the base station. Notice that the cost is a non-increasing function of time, as predicted in the proof of Theorem~\ref{thm:convergence}, eventually settling as the coverage regions/generators reach their limiting configuration, e.g., see Figure~\ref{fig:Base Simulation}. These configurations are guaranteed to be Pareto optimal and, by Remark~\ref{rem:centroidal}, to form a multiplicatively weighted Voronoi partition. Since vehicles are assumed to have identical maximum speeds, we see from the limiting configuration in Figure~\ref{fig:Base Simulation} that the resultant coverage assignments provide load-balancing that takes into account the event likelihood. If the low-level trajectory planner biases trajectories according to the event likelihood, this results in desirable coverage properties. Under the modified SMC planner used here, the temporal distribution of agent locations closely resembles the spatial likelihood distribution in the limit, as shown in Figure~\ref{fig:dsmc}.
\begin{figure} [h]
	\centering
	\includegraphics[width = 0.6\columnwidth]{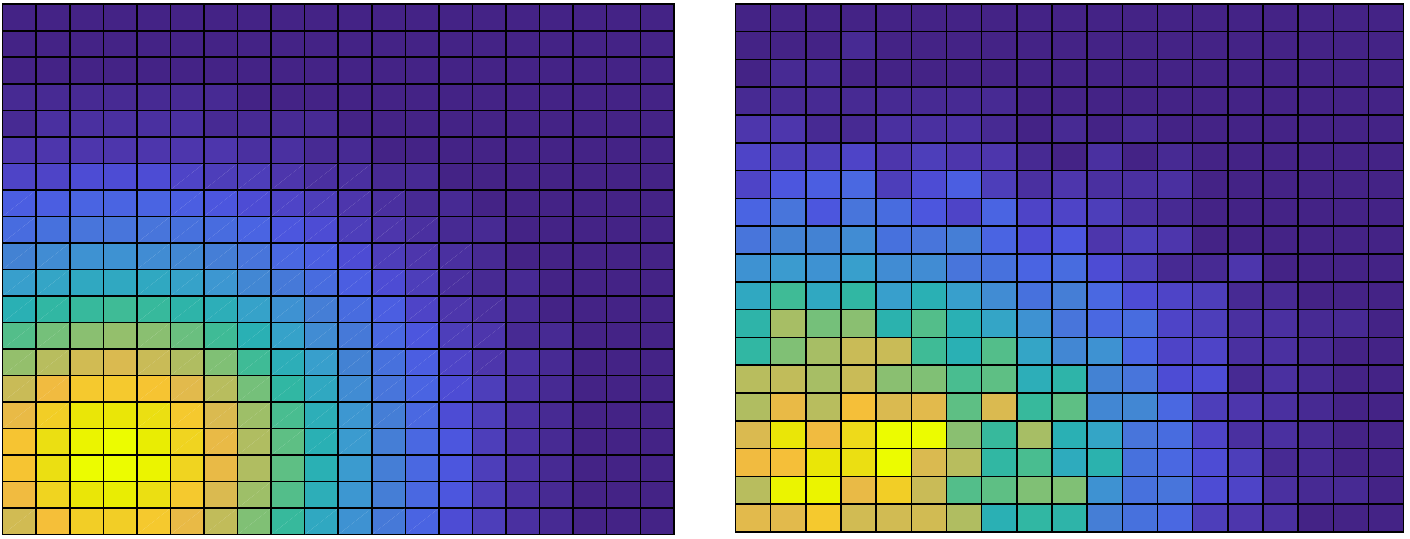}
	\caption{Comparison between the (static) event likelihood $\Phi$ (left), and the proportion of time that some agent occupied each subregion after significant time has passed (10000 units) (right).}
	\label{fig:dsmc}
\end{figure}

Further, during the simulation, no two agents ever occupied the same space due to the careful parameter manipulations employed by Algorithm~\ref{alg:partitioning}. Figure~\ref{fig:colocation} illustrates the logic governing these manipulations through a simplistic example. 
\begin{figure} [h]
	\centering
\includegraphics[width = 0.85\columnwidth]{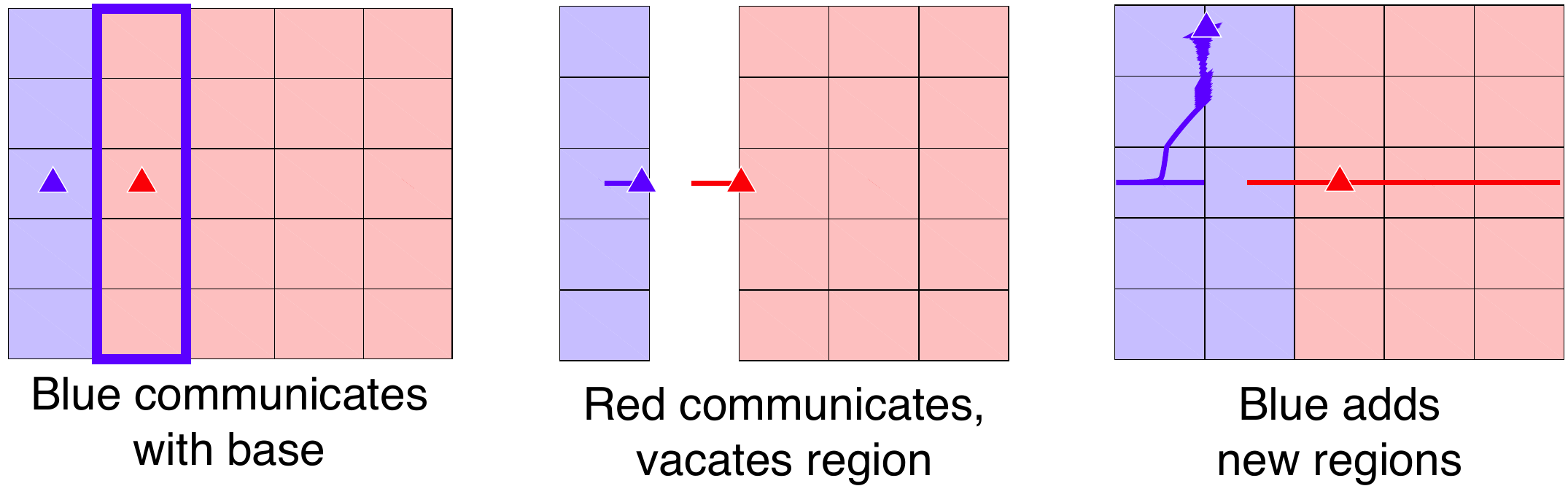}
	\caption{Simplified example illustrating how Algorithm~\ref{alg:partitioning} manipulates timing parameters to prevent agent collisions.}
	\label{fig:colocation}
\end{figure}
During the first update, the blue agent acquires some of the red agent's coverage region. Rather than immediately adding these regions to its active covering, the blue region waits until sufficient time has passed to guarantee that the red agent has updated and moved out of the reassigned regions. Under Algorithm~\ref{alg:complete}, once the red agent communicates with the base, it immediately vacates the re-assigned regions, after which the blue agent can add the region to its active covering. This procedure guarantees that no two agents will never have overlapping active coverings and thus never collide (Theorem~\ref{thm:collision}). This same logic is results in inherent collision prevention over more complex scenarios.

\subsection{Quasi-static Likelihood}

We now illustrate how the proposed coverage framework reacts to abrupt changes in the underlying likelihood, i.e., when the likelihood is quasi-static. This type of scenario is common in realistic missions, e.g., when the base-station's estimate of the underlying likelihood is only re-formulated if some agent's sensor data indicates a drastic change in the underlying landscape. For this purpose, we adopt identical parameters as in the first example, with the exception of the likelihood $\Phi$, whose spatial distribution abruptly switches at select time-points. If the switches are sufficiently spaced in comparison to the rate of convergence, then the coverage regions dynamically adjust to an optimal configuration that is reflective of the current state. For example, Figure~\ref{fig:agentschangingDist} shows the coverage region evolution after the underlying likelihood undergoes a single switch from the initial to the final density shown in Figure~\ref{fig:likelihood} at time $t = 2000$.
\begin{figure} [h]
	\centering
	\includegraphics[width = 0.49\columnwidth]{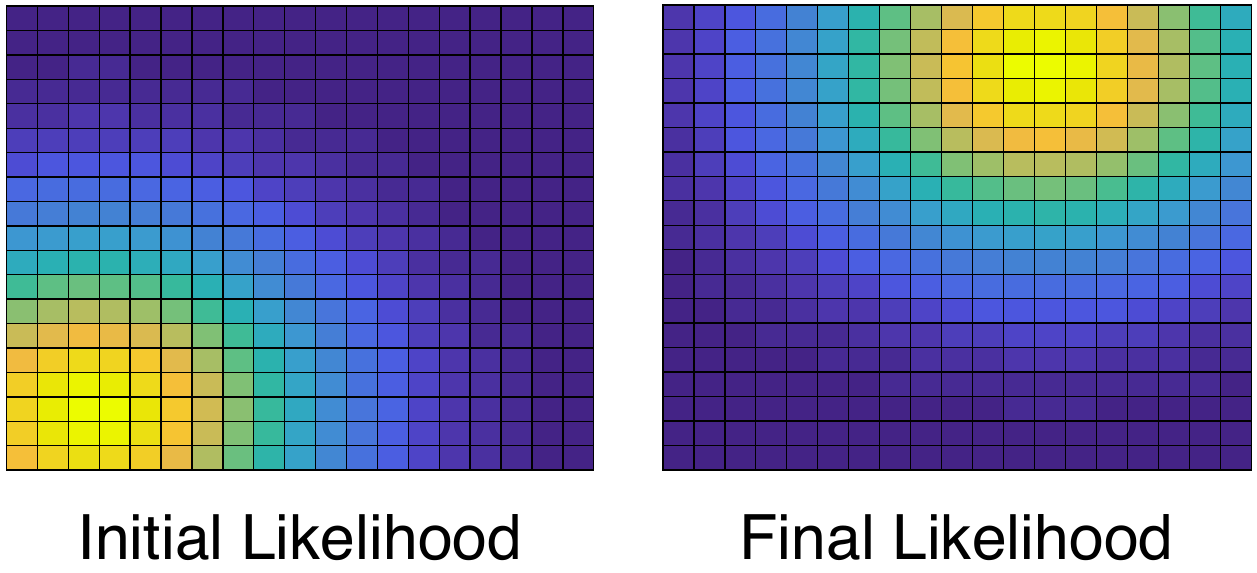}
	\caption{The initial and final likelihood $\Phi(\cdot, t)$.}
	\label{fig:likelihood}
\end{figure}
\begin{figure} [h]
	\centering
	\includegraphics[width = 0.99\columnwidth]{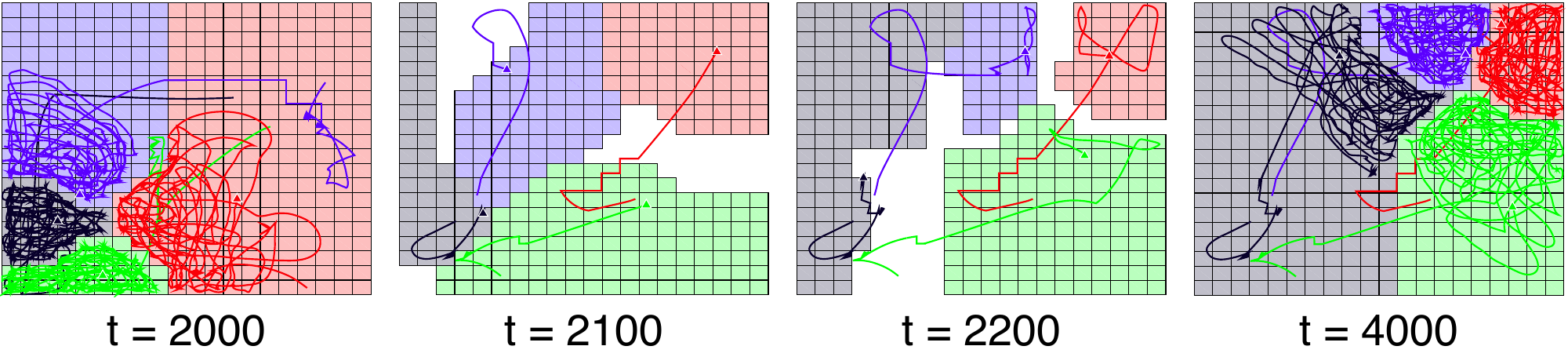}
	\caption{Coverage regions after the likelihood switches (see Fig.~\ref{fig:likelihood})}
	\label{fig:agentschangingDist}
\end{figure}
In contrast, when the underlying likelihood changes faster than the rate of convergence, coverage regions are constantly in a transient state. Despite this, the proposed framework still results in some degree of load-balancing. To illustrate, Figure~\ref{fig:CoverageChangingDists} shows the value of the cost $\mathcal{H}$ during a simulation in which the underlying likelihood switches at $12$ randomly chosen time-points over a $1000$ unit horizon. Each switch re-defined the spatial likelihood as a Gaussian distribution centered at a randomly selected location. Notice that the value of the cost function monotonically decreases between the abrupt spikes caused by changes in the underlying likelihood. Despite the fact that convergence is not reached, coverage regions quickly shift away from high-cost configurations, as indicated by the sharp decreases in the cost shortly following each switch.
\begin{figure} [h]
	\centering
	\includegraphics[width = 0.6\columnwidth]{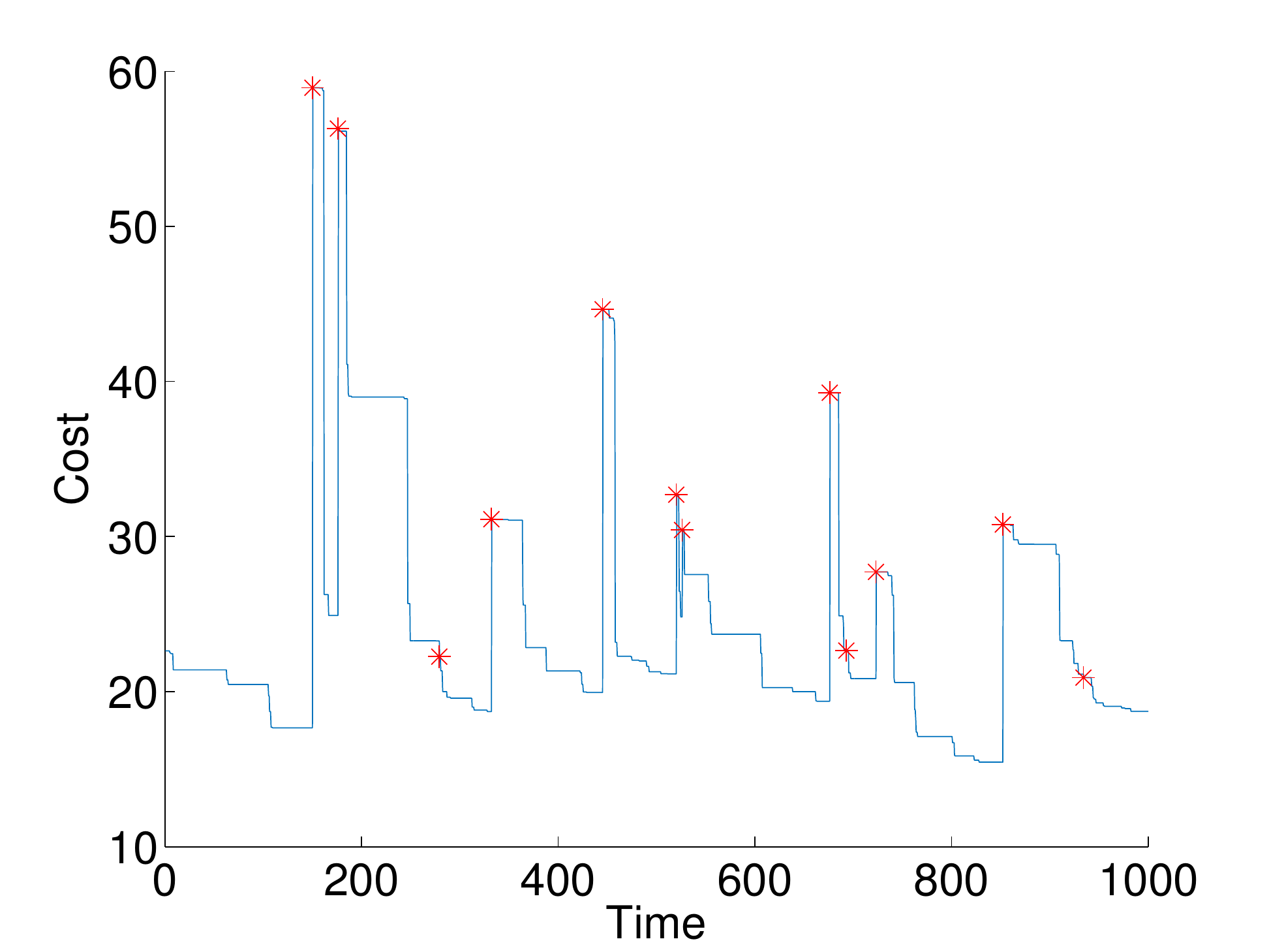}
	\caption{Evolution of the cost $\mathcal{H}$ using a quasi-static likelihood with $12$ random switches (switch times are indicated by the stars).}
	\label{fig:CoverageChangingDists}
\end{figure}

\section{Conclusion}
\label{sec:Conclusion}
This work presents a modular, decomposition-based, coverage control framework for communication-constrained multi-agent surveillance missions.
In particular, our approach uses a dynamic partitioning strategy to balance the surveillance load across available agents, requiring only sporadic and unplanned exchanges between individual agents and a base station. The partitioning update algorithm also manages high-level timing and logic parameters to guarantee that the resulting coverage assignments have geometric and temporal properties that are amenable for combination with generic single vehicle trajectory planners. Under appropriate assumptions, the proposed algorithms will produce a set of coverage regions forming a Pareto optimal partition, while also ensuring collision avoidance and quality of coverage guarantees.

Future work should further relax communication assumptions to reflect additional realistic hardware limitations, e.g., use of directional antennae for wireless transmission. Other areas of future research include the addition of peer-to-peer, in addition to central, communication, performance comparisons between specific trajectory planners when used within our framework, e.g., those involving ergodic Markov chains, and  further theoretical characterizations of performance.

\section*{Appendix: Proofs}
\begin{proposition}[Set Membership] 
Suppose Assumption~\ref{assn:initialization} holds, and that, at the time of each exchange occurring prior to the fixed time $\bar{t}\geq 0$, required algorithmic constructions are well-posed so that the base station and the communicating agent are able to update their respective variables via Algorithm~\ref{alg:partitioning}. Then, for any $k \in Q$ at any time $t \leq \bar{t}$:
\begin{enumerate}
\item $k \in P_{\text{ID}_k}$,
\item $k$ belongs to at most $2$ elements of $P$,
\item if $T_{\text{ID}_k} = 0$, then $k\notin P_j$ for any $j \neq \text{ID}_k$, and
\item if $k \in P_j$, $j\neq \text{ID}_k$, then $P_j \cap P^{\mathcal{ID}}_\ell = \emptyset$ for $\ell \notin\{j, \text{ID}_k\}$
\end{enumerate}
\label{prop:set}
\end{proposition}
\renewcommand*{\proofname}{Proof}
\begin{proof}
Fix $\bar{t}\geq 0$, $k \in Q$. When $t = 0$, $P = P^{\mathcal{ID}}$ is an $m$-partition of $Q$, implying the proposition. 
Since $k$ is not removed from $P_{\text{ID}_k}$ or added to any $P_{i}$ with $i \neq \text{ID}_k$ until its first \emph{reassignment}, i.e., when $\text{ID}_k$ is changed. Thus, the proposition is true for all $t$ prior to the first reassignment.  Suppose the proposition holds for all $t$ prior to the $p^{\text{th}}$ reassignment, which occurs at $t = t_0$. Suppose that $\text{ID}^-_k = j$, $\text{ID}^+_k = i \neq j$. Algorithm~\ref{alg:partitioning} defines $P_i^{{\mathcal{ID}},+} = P^+_i = P^+_{\text{ID}^+_k}$. Thus, $k \in P^+_{\text{ID}^+_k} = P^+_i$ and remains in these sets until another reassignment. Statement 1 holds, therefore, for all $t$ prior to the $p+1^{\text{st}}$ reassignment.
Now note that, by Alg.~\ref{alg:timer}, reassignment cannot occur at $t_0$ unless $T_j^- = 0$. By inductive assumption, statement 3 of the proposition holds when $t = t_0^-$, implying that $k \notin P^-_\ell$ for any $\ell \neq j$. Upon reassignment, the timers $T_j, T_i$ are modified such that $T^+_j, T^+_i>\omega^+_j+\overline{\Delta}-t_0$. 
Since 1) $\text{ID}_{k}$ is unchanged as long as $T_j > 0$, and 2) agent $j$ communicates with the base prior to time $\omega^+_j+\overline{\Delta}$ (when $k$ is removed from $P_j$), we deduce that $k$ solely belongs to $P_j$ and $P_i$ until the $p+1^{\text{st}}$ reassignment.
By the same logic, at any time $t\geq t_0^+$ at which $T_i = 0$ and the $p+1^{\text{st}}$ reassignment has not yet occurred, $k \in P_i$ exclusively (addition to other sets in $P$ without reassignment is impossible). We deduce that statements 2 and 3 hold for any $t$ prior to the $p+1^{\text{st}}$ reassignment. Now note that, 1) statement 3 holds prior at time $t_0^+$, implying $P_j^- = P^{\mathcal{ID},-}_i$, 2) no agent claims additional points vertices from $P_j^+$ unless $T_j = 0$, and 3) $k$ cannot be added to a coverage region without reassignment. As such, $P_j \cap P^{\mathcal{ID}}_\ell = \emptyset$ for any $\ell \notin\{j, \text{ID}_k = i\}$ prior to the $p+1^{\text{st}}$ reassignment. From persistent communication and the bound $\underline{\Delta}$, the proposition follows by induction.
\end{proof}

\renewcommand*{\proofname}{Proof of Theorem~\ref{thm:well-posedness}}
\begin{proof}
It suffices to show that Def.~\ref{def:P_plus} is well-posed (Prop.~\ref{prop:well-posed}) whenever additive sets are required. We proceed by induction.
When $t = 0$, $P^\mathcal{ID} = P$ is a connected $m$-partition of $Q$ and thus, for any $i$, $P_i^{\mathcal{ID}}$ is disjoint from $\bigcup_{j \neq i} P_j$. The same holds prior to the first exchange between the base station and some agent. Thus, by Prop.~\ref{prop:well-posed}, the first call to Alg.~\ref{alg:partitioning} is well-posed. 
Now assume that, for all times $t$ prior to the $p^{\text{th}}$ call to Alg.~\ref{alg:partitioning}, 1) $P^\mathcal{ID}$ is a connected $m$-partition of $Q$, and 2) if an exchange that requires construction of  $\supscr{P_i}{add}$ occurs, then $P_i^{\mathcal{ID}}\cap\left(\bigcup_{j \neq i} P_j\right) = \emptyset$. This implies that Prop.~\ref{prop:set} holds at any time $t$ prior to the $p+1^{\text{st}}$ exchange. 
Assume the $p^{\text{th}}$ communication occurs when $t = t_0$ and involves agent $i$.
By the uniqueness of identifiers, $P^{\mathcal{ID},+}$ is an $m$-partition of $Q$. To show that $P^{\mathcal{ID},+}$ is connected, first note that $P_i^{\mathcal{ID},+} = P_i^+$. Since $P_i^+ = \supscr{P_i}{add}(c_i^+)$ (connected by Def.~\ref{def:P_plus}), or $P_i^{\mathcal{ID},+} = P^{\mathcal{ID},-}_i$ (connected by inductive assumption) connectivity of $P_i^{\mathcal{ID},+}$ follows. 
Now consider $P_j^{\mathcal{ID}}$, $j \neq i$. If $T_j^- \neq 0$, then $P^{\mathcal{ID},-}_j = P_j^{\mathcal{ID},+}$ and connectivity of $P_j^{\mathcal{ID},+}$ follows. 
Suppose $T_j^- = 0$ and $P_j^{\mathcal{ID}, +}$ is not connected.
Then, $P_j^+$ is not connected: if it were, there would exist $k \in P_j^+$ with $\text{ID}^+_k \notin\{i,j\}$, contradicting Prop.~\ref{prop:set}, satement. Thus, there exists $k_1 \in P_j^{\mathcal{ID}, +}$ such that 1) $k_1 \notin \supscr{P_i}{add}(c_i^+)$, and 2) any length-minimizing path in $G(P^{-}_j)$ spanning $k_1$ and $c^+_j$ contains some $k_2\in \supscr{P_i}{add}(c_i^+)= P_i^+(c_i^+)$. Select one such path and vertex $k_2$. Assume without loss of generality that there exists an edge $\{k_1,k_2\} \in \mathcal{E}$. Def.~\ref{def:P_plus} implies $\frac{1}{s_i}d_{P_i^+}(k_2, c_i^+) < \min\{\frac{1}{s_\ell}d_{P^+_\ell}(k_2, c^+_\ell)|\ell \neq i, k_2 \in P^+_i\}$ and thus $\frac{1}{s_i}d_{P_i^+\cup\{k_1\}}(k_1, c_i^+) < \frac{1}{s_j}d_{P^+_j}(k_1, c^+_j)$. Since $T^-_j = 0$ and $\text{ID}^-_{k_1} = j$, Prop.~\ref{prop:set} implies $\frac{1}{s_i}d_{P_i^+\cup\{k_1\}}(k_1, c_i^+) <\frac{1}{s_j}d_{P^+_j}(k_1, c^+_j) = \min\{\frac{1}{s_\ell}d_{P^+_{\ell}}(k_1, c^+_\ell)|\ell \neq i, k_1 \in P_\ell^+\}$, contradicting $k_1 \notin \supscr{P_i}{add}(c_i^+)$. Thus, $P_j^{\mathcal{ID},+}$ is connected. Invoking Prop.~\ref{prop:set} statement 3, the inductive assumption therefore holds for all times prior to the $p+1^{st}$ exchange, thereby implying well-posedness of the first $p+1$ exchanges.
\end{proof}

\renewcommand*{\proofname}{Proof of Theorem~\ref{thm:set_properties}}
\begin{proof}\hfill\\
\underline{\emph{Statement 1}}: The proof of Thm.~\ref{thm:well-posedness} implies the statement.

\noindent\underline{\emph{Statement 2}}: $P$ is an $m$-covering of $Q$ since $P^{\mathcal{ID}}$ is always an $m$-partition of $Q$ (statement 1), and $P_i^{\mathcal{ID}} \subseteq P_i$ for any $i$ (Prop.~\ref{prop:set}, statement 1). $P$ is connected, since $P_i=P_i^{\mathcal{ID}}$ (connected by statement 1) immediately following any update, and $P_i$ is unchanged in between updates. 

\noindent\underline{\emph{Statement 3}}: It suffices to show that $\text{ID}_{c_i} = i$  for any $t$ and any $i$: this would imply $c_i \neq c_j$ for any $i \neq j$, and  $c_i \in P_i$ (Prop.~\ref{prop:set}).  By assumption, $\text{ID}_{c_i} = i$ for all $i$ at $t = 0$. The same holds for any $t$ prior to the first exchange between \emph{any} agent and the base station. Suppose $\text{ID}_{c_i} = i$ for all $i$ (thus $c_i \neq c_j$ for any $i \neq j$) prior to the $p^{\text{th}}$ exchange. If agent $i$ is the $p^{\text{th}}$ communicating agent, lines $2$ and $9$ of Alg.~\ref{alg:partitioning} imply $\text{ID}^+_{c_i^+} = i$. Since $d_{P^-_j}(c^-_j, c^-_j) = 0$ for any $j$, we have $c^+_j \notin \supscr{P_i}{add}(c_i^+)$. Thus, $\text{ID}^+_{c^+_j} = j$, and induction proves the statement.

\noindent\underline{\emph{Statements 4 and 5}}: Statement 4 follows from~\eqref{eqn:phi_temp}, noting that $P_i^A = P_i$.
Statement 5 holds by assumption when $t = 0$. 
Let $k \in Q$, and consider times when $\text{ID}_k$ changes ($k$ is \emph{re-assigned}). Since $\text{supp}(\Phi^A_j(\cdot, t)) = P_j = P^A_j$ for any $j$ at $t = 0$, statement 4 implies that, for any $t$ prior to the first reassignment, $k$ belongs exclusively to $\text{supp}(\Phi^A_{\text{ID}_k}(\cdot, t))$, implying statement 5.
Suppose statement 5 holds for all $t$ prior to the $p^{\text{th}}$ reassignment (occurring at time $t_0$), and $\text{ID}^-_k = j$, $\text{ID}^+_k = i\neq j$. Then, $T_j^- = 0$ and $k$ belongs exclusively to $P_j^-$ when $t = t_0^-$ (Prop.~\ref{prop:set}). By Alg.~\ref{alg:partitioning} and~\ref{alg:timer}, $k \in \supscr{P_i}{$A$,pd,+}$ and $T_i^+ > \omega^+_j +\overline{\Delta}-t_0\geq \tau^{A,+}_i$. Since $\text{supp}(\Phi^A_i(\cdot, t))$ is not redefined for a duration of at least $T^+_i\geq \tau_i^{A,+}$, ~\eqref{eqn:phi_temp} implies $k\notin\text{supp}(\supscr{\Phi_i}{A}(\cdot, t))$ when $t \in [t_0^+, t_0^++\tau^{A,+}_i]$. Since $k$ is re-assigned when $t =t_0$, $k \in P_i^+\backslash P^{\mathcal{ID},-}_i$ and $T_j^+ = \omega^+_j +\overline{\Delta}-t_0$. Agent $j$ will communicate with the base at some time $t_1<t_0+T_j^+ = \omega_j^++\overline{\Delta}<t_0+T_i^+$. Thus, $T_i>0$ when $t = t_1$, and $k$ is removed from both $P_j$ and $\text{supp}(\Phi_j^A(\cdot, t))$. Thus, for all $t >t_0+\tau^+_i$ and before the $p+1^{\text{st}}$ reassignment, $k$ belongs exclusively to  $\text{supp}(\phi_i(\cdot, t))$. 
\end{proof}
\renewcommand*{\proofname}{Proof of Theorem~\ref{thm:characteristics}}
\begin{proof}Thm.~\ref{thm:set_properties} implies statement 1. 

\noindent\underline{\emph{Statement 2}}: For any $i$,  1) $T_i = 0$ when $t = 0$, and 2) $\frac{1}{s_i}d_{Q'}(k_1,k_2) \leq \overline{d}$ for any $Q' \subseteq Q$, $k_1,k_2 \in Q'$. Thus, it is straightforward to show that, for any $i$ and any $T$, we have the bound $T_i \leq\overline{\Delta}+\subscr{\Delta}{H}+\overline{d}$. We show that, for any $i$, the bound $\tau_i^A - t+ \omega^A_i \leq T_i - \subscr{\Delta}{H}$ holds by induction: $T_i = 0$ and $\tau_i^A = -\subscr{\Delta}{H}$ when t = 0 , so $\tau_i^A - t+ \omega^A_i = \tau_i^A \leq T_i-\subscr{\Delta}{H}$. The bound similarly holds prior to the first exchange involving \emph{any} agent, since $\tau_i^A - t = \tau_i^A - t+\omega^A_i \leq -\subscr{\Delta}{H}\leq T_i-\subscr{\Delta}{H}$ at any such time. Assume the bound holds prior to the $p^{\text{th}}$ update (occuring at $t = t_0$). Consider $2$ cases: if agent $i$ is the communicating agent, then $\tau_i^{A,+} - t+ \omega^A_i  = \tau^+_i \coloneqq T^+_i-\subscr{\Delta}{H}$; if not, then $\tau_i^{A,+} = \tau_i^{A,-}$ and either 1) $T_i^- = T_i^+$ implying the desired bound, or 
2) $T_i^- = 0$ and $\tau^{A,+}_i -t_0+ \omega^{A,+}_i = \tau^{A,-}_i - t+ \omega^{A,-}_i  \leq T_i^--\subscr{\Delta}{H} = -\subscr{\Delta}{H} \leq (\omega_i^{A,+}+\overline{\Delta}-t_0)  -\subscr{\Delta}{H} = T_i^+ -\subscr{\Delta}{H}$.
This logic extends to all times prior to the $p+1^{\text{st}}$ exchange and the desired bound follows by induction.

Using the bounds from the previous paragraph, we have $\tau_i^A + \omega_i^A \leq t+\overline{\Delta}+\overline{d}$. Fix $t$ and $k \in \text{Proh}_i(t)$. Then, $k \in P^{A,+}_i = P^+_i$, $k \in \supscr{P_i}{$A$,pd,+}$, and $t-\omega^{A,+}_i < \tau_i^{A,+}$ (`$+$' is with respect to the fixed time $t$).
Further, over the interval $[t,\omega^{A,+}_i+\tau_i^{A,+}]$, the vertex $k$ is not re-assigned, $P_i$ is not augmented, and $\tau_i^A$ is unchanged . 
Therefore, $k\notin\text{Proh}_i(t)$ at time $\omega^{A,+}_i+\tau^{A,+}_i$. 
If $t_0 \coloneqq \omega^{A,+}_i+\tau^{A,+}_i$, we have $t<t_0 \leq t+\overline{\Delta}+\overline{d}$. Since $T_i \geq \tau_i^A+\subscr{\Delta}{H}$ at time $\omega^{A,+}_i$, $k$ is not re-assigned during the interval $[\omega^{A,+}_i, \omega^{A,+}_i+T^+_i] \supseteq [\omega^{A,+}_i, t_0+\subscr{\Delta}{H}] \supseteq [t_0, t_0+\subscr{\Delta}{H}]$. Thus $k \in P_i\backslash \text{Proh}_i (\cdot)$ over the same interval.


\noindent\underline{\emph{Statement 3}}: Fix $t$ and suppose $k \in P_i^-\backslash P_i^+$ (in this proof, `$+,-$' are with respect to $t$). Then, 1) $\text{ID}_k$ changed ($k$ was reassigned) at time $t_0  < t$, 2) agent $i$ communicates with the base at time $t$, and 3) no exchanges involving agent $i$ occurred during the interval [$t_0,t)$. Upon reassignment at time $t_0$, Alg.~\ref{alg:timer} specifies that 1) $T_i$ is reset to value $\omega_i^{A,-}+\overline{\Delta}-t_0$, thus $P^{\mathcal{ID}}_i$ is unchanged over the interval $[t_0, t)$, 2) $k$ is added to  $\supscr{P_{\text{ID}_k}}{$A$,pd}$, and 3)  $\tau^A_{\text{ID}_k}$, $T_{\text{ID}_k}$ are given values of at least
\begin{equation*}
\tilde{\omega} \coloneqq \underset{\tilde{k} \in P_i^-\backslash P_i^+}{\max}\left\{\omega_i^{A,-}+\overline{\Delta}+\frac{1}{s_i}d_{P^-_i}\left(\tilde{k}, P_i^{\mathcal{ID},-}\right) - t_0\right\},
\end{equation*}
implying that $P_{\text{ID}_k}$, $\text{Proh}_{\text{ID}_k}(\cdot)$ remain unchanged over the interval $(t_0, \tilde{\omega}] \supseteq (t_0, t+ \frac{1}{s_i}d_{P^-_i}(k, P^{\mathcal{ID},-}_i)]\supseteq (t_0, t]$. 

Since coverage regions are connected and non-empty (Thm.~\ref{thm:set_properties}) and $P_i^-\cap P^{\mathcal{ID}}_\ell  = \emptyset$ for any $\ell \notin \{i, \text{ID}^+_k\}$ on the interval $(t_0, t]$ (Prop.~\ref{prop:set}), 1) there exists a path of length $d_{P_i^-}(k, P_i^{\mathcal{ID},-})$ from $k$ into $P_i^{\mathcal{ID},-}$ and every vertex along any such path (except the terminal vertex) lyies within $P_i^- \backslash P_i^+$, and 2) $P_i^- \backslash P_i^+  \subseteq \text{Proh}_{\text{ID}_k^+}$ over the interval $(t_0, \tilde{\omega}] \supseteq [t,  t+ \frac{1}{s_i}d_{P^-_i}(k, P^{\mathcal{ID},-}_i)]$. Since 1) each vertex belongs to no more than two sets simultaneously (Prop.~\ref{prop:set}), 2) $k \in P_i^-\backslash P_i^+$, and 3) no agent claims any vertex in the set $\text{Proh}_{\text{ID}_k^+}$ when $T_{\text{ID}_k^+}>0$, vertices along the path (excluding the terminal vertex) do not belong to $P_j$ with $j \neq \text{ID}_k$ over the interval $[t,  t+ \frac{1}{s_i}d_{P^-_i}(k, P^{\mathcal{ID},-}_i)]$. 
To complete the proof, note that Alg.~\ref{alg:timer} implies $T_i^+ > \frac{1}{s_i}d_{P^-_i}(k, P^{\mathcal{ID},-}_i)$, and thus $P_i^{\mathcal{ID},-} \subseteq P_i^{\mathcal{ID}}$ over $[t,  t+ \frac{1}{s_i}d_{P^-_i}(k, P^{\mathcal{ID},-}_i)]$.
\end{proof}
\begin{proposition}[Cost]
Suppose Assumption~\ref{assn:initialization} holds and that, upon each exchange, the base station and the communicating agent update their respective variables via Alg.~\ref{alg:partitioning}.
If $\Phi(\cdot,t_1) = \Phi(\cdot, t_2)$ for all $t_1,t_2$, then
$\mathcal{H}(c, P^{\mathcal{ID}}, \cdot) = \mathcal{H}(c, P, \cdot)$.
\label{prop:H_min}
\end{proposition}
\renewcommand*{\proofname}{Proof}
\begin{proof}
Since $\Phi$ is static, $\mathcal{H}(\cdot, \cdot, t_1) = \mathcal{H}(\cdot, \cdot, t_2)$ for any $t_1, t_2$. When $t = 0$, $P = P^{\mathcal{ID}}$ and thus $\mathcal{H}(c, P^{\mathcal{ID}}, 0) = \mathcal{H}(c, P, 0)$. The same is true prior to the first exchange between \emph{any} agent and the base. Suppose that immediately prior to the $p^{\text{th}}$ exchange (occurring at $t = t_0$, involving agent $i$), we have $\mathcal{H}(c^-, P^{\mathcal{ID}, -}, t_0) = \mathcal{H}(c^-, P^-, t_0)$. Recall that, for any agent $j$,  $P_j$ and $P^{\mathcal{ID}}_j$ coincide immediately following any exchange involving agent $j$ and, if agent $j$ claims vertices from $P_i$, then Alg.~\ref{alg:timer} guarantees that agent $i$ will communicate with the base station before any additional vertices are claimed by other agents. At the time of the $p^{\text{th}}$ update, this logic, along with Prop.~\ref{prop:set}, implies that $P^{\mathcal{ID},-}_i\cap P_j^- = \emptyset$, for all $j \neq i$. Noting that $c_i^+ \in P_i^{\mathcal{ID},-}$, we deduce that any single $k \in P^{\mathcal{ID},-}_i$ contributes equally to $\mathcal{H}(c^+, P^{\mathcal{ID},+}, t_0)$ and $\mathcal{H}(c^+, P^+, t_0)$.  If $k \in \supscr{P_i}{add}(c_i^+)\backslash P_i^{\mathcal{ID},-}$, then for any $j\neq i$ such that $k \in P^+_j$, we have $\frac{1}{s_i}d_{P_i^+}(k, c_i^+) < \frac{1}{s_j}d_{P_j^+}(k, c^+_j)$ (Def.~\ref{def:P_plus}). Therefore, $k$ contributes equivalently to the values of both  $\mathcal{H}(c^+, P^{\mathcal{ID},+}, t_0)$ and $\mathcal{H}(c^+, P^+, t_0)$. 
Now suppose $k \in P_j^+ \backslash P_i^+$, where $P_j^+ \cap P^+_i \neq \emptyset$. We show that $d_{P_j^{\mathcal{ID},+}}(c_j^+, k)  = d_{P_j^{+}}(c_j^+, k)$: if a length-minimizing path in $G(P_j^+)$ between $c_j^+$ and $k$ is also contained in $G(P_j^{\mathcal{ID},+})$, then the result is trivial. Suppose that every such minimum length path does leave $G(P_j^{\mathcal{ID},+})$. From Prop.~\ref{prop:set} statement 4, every $\bar{k} \in P_j^+$ must satisfy either $\text{ID}^+_{\bar{k}} \in  \{i, j\}$. 
Therefore, we assume without loss of generality that $k$ is adjacent to $P_i^+$. Let $\overline{k} \in P_i^+$ be a vertex that is adjacent to $k$ and lies along a minimum-length path in $G(P_j^+)$ spanning $c_j^+$ and $k$. Since $\overline{k} \in P_i^+\backslash P_i^{\mathcal{ID},-}$, we must have $\overline{k} \in \supscr{P_i}{add}(c_i^+)$ as constructed during the update, which implies $\frac{1}{s_i}d_{P_i^+}(\overline{k}, c_i^+) < \min\{\frac{1}{s_\ell}d_{P^+_\ell}(\overline{k}, c^+_\ell)|\ell \neq i, \bar{k} \in P_\ell^+\}$ and thus $\frac{1}{s_i}d_{P_i^+\cup\{k\}}(k, c_i^+) < \frac{1}{s_j}d_{P^+_j}(k, c^+_j)$. Since $T^-_j = 0$ and $\text{ID}^-_{k} = j$, Prop.~\ref{prop:set} implies $\frac{1}{s_i}d_{P_i^+\cup\{k\}}(k, c_i^+) <\frac{1}{s_j}d_{P^+_j}(k, c^+_j) = \min\{\frac{1}{s_\ell}d_{P^+_{\ell}}(k, c^+_\ell)|\ell \neq i, P_\ell^+\}$, contradicting $k \notin \supscr{P_i}{add}(c_i^+) \subset P_i^+$. Thus,  $d_{P_j^{\mathcal{ID},+}}(c_j^+, k)  = d_{P_j^{+}}(c_j^+, k)$, which, by inductive assumption, implies that $k$ contributes equally to the value of  both  $\mathcal{H}(c^+, P^{\mathcal{ID},+}, t_0)$ and $\mathcal{H}(c^+, P^+, t_0)$. We conclude that $\mathcal{H}(c^+, P^{\mathcal{ID},+}, t_0) = \mathcal{H}(c^+, P^+, t_0)$. Since $P$, $P^{\mathcal{ID}}$, and $c$ are static between updates, the statement follows by induction.
\end{proof}

\renewcommand*{\proofname}{Proof of Theorem~\ref{thm:convergence}}
\begin{proof}
The value of $\mathcal{H}(c,P,t)$ is static in between base station exchanges since $P$ and $c$ do not change in between updates. 
Consider an update occurring at $t = t_0$ involving agent $i$. Noting Prop.~\ref{prop:H_min}, we have $\mathcal{H}(c^+, P^+, t_0)\leq \mathcal{H}(c^-, P^{\mathcal{ID},-}, t_0) = \mathcal{H}(c^-, P^-, t_0)$. Thus, the value $\mathcal{H}(c, P, t)$ is non-increasing as $t\to \infty$. 
Since $\text{Cov}_m(Q)$ is finite, there must exist some time $t_0$ after which the value of  $\mathcal{H}$ is static. Consider fixed $t>t_0$ at which some agent $i$ communicates with the base. Since the value of $\mathcal{H}$ does not change during the update, Alg.~\ref{alg:partitioning} implies that $P^{\mathcal{ID}}$ and $c$ will be unchanged by the update. It follows that  $c$ and $P^{\mathcal{ID}}$ converge in finite time. Further, since $P^{\mathcal{ID}}_i \subseteq P_i$ for any $i$ (Prop.~\ref{prop:set}), we have $P_i^{\mathcal{ID},-} = P_i^{\mathcal{ID},+} = P^{\mathcal{ID},-}_i \cup \supscr{P_i}{add}(c_i^+) = P_i^+$. Noting the persistence of communication imposed by $\overline{\Delta}$, the same logic implies that after some finite time, $P$ and $P^{\mathcal{ID}}$ are concurrent.
 
We show that the limiting configuration is Pareto optimal. Consider $t_0$, such that for all $t>t_0$, $c$ and $P$ are static and $P$ is an $m$-partition of $Q$. Since timers $T_i$ are only reset when $P$ is altered, we assume without loss of generality that $T_i = 0$ for all $i \in \until m$ at any $t>t_0$. Suppose that agent $i$ communicates with the base station at time $t>t_0$. With these assumptions, Alg.~\ref{alg:partitioning} implies that there exists no $k \in  P_i$ such that $\sum_{h\in P_i} d_{P_i}(h,k)\Phi(h,t) < \sum_{h\in P_i} d_{P_i}(h,c_i)\Phi(h,t)$ (otherwise the value of $\mathcal{H}$ would be lowered by moving $c_i$). 

Similarly, for any $k \in P_j$ with $j \neq i$ that is adjacent to $P_i$, we must have $\frac{1}{s_i}d_{P_i\cup\{k\}}(c_i, k) \geq \frac{1}{s_j}d_{P_j}(c_j, k)$. Indeed, if this were not so, there would exist $k \in \supscr{P_i}{add}(c_i^+)\backslash P_i^-$, contradicting the convergence assumption, since $\supscr{P_i}{add}(c_i^+) = P_i^+$. As such, for any $i$, there is no $Q' \subset Q\backslash P_i$ such that $\sum_{k \in Q'} \frac{1}{s_i}d_{P_i \cup Q'}(c_i, k) < \sum_{k \in Q'} \min \{\frac{1}{s_j}d_{P_j}(c_j, k) | k \in P_j, j \neq i\}$, which implies statement (ii) of Def.~\ref{def:pareto}. 
\end{proof}
\renewcommand*{\proofname}{Proof of Theorem~\ref{thm:collision}}
\begin{proof}
By Assumption~\ref{assn:motion}, no agent ever leaves its assigned coverage region or enters its prohibited region, provided no abrupt changes to these regions occur during an update. Therefore, if no update ever occurs in which the vertex corresponding to the communicating agent's location is removed from the relevant agent's coverage region, then the statement is immediate. Suppose now that, at some time $t = t_0$, agent $i$,  whose location is associated with some $k \in P^{A,-}_i$, communicates with the base station and $k$ is removed, i.e., $k \notin P^{A,+}_i$. At time $t^+$, agent $i$ executes lines $5$ and $6$ of Alg.~\ref{alg:complete}. Thm.~\ref{thm:characteristics}, however,  guarantees that 1) there will exist a path in $G(P^{A,-}_i)$ between $k$ and the set $P_i^{\mathcal{ID},-}$, 2) that all vertices along this path belong to $\text{Proh}_{\text{ID}_k^+}\backslash \bigcup_{j \neq \text{ID}^+_k} P^+_j$ during the time period $(t_0, t_0+\frac{1}{s_i}d_{P^{A,-}_i}(k,P^{\mathcal{ID},-}_i)]$, and 3) $P_i^{\mathcal{ID},-} \subseteq P_i \coloneqq P_i^A$ over the same interval. Therefore, if agent $i$ immediately starts moving along the path, its location will lie exclusively within $\text{Proh}_{\text{ID}_k^+}$ until it reaches $P_i^{A,+}$. 

It remains to show that no agent $i$ ever enters $\text{Proh}_i(t)$. Whenever agent $i$ is executing lines $1-2$ of Alg.~\ref{alg:complete}, it follows readily that it will not enter $\text{Proh}_i(t)$. We  show that the same holds when agent $i$ is forced to execute line $5$ and $6$ of Alg.~\ref{alg:complete}. Without loss of generality, consider the update at time $t_0$ previously described. Since $k$ is re-assigned prior to the update at time $t_0$, we have $\text{Proh}^-_i = \emptyset$ (since vertices in $P_i$ cannot be claimed unless $T_i = 0$, implying $t_0 - \omega_i^{A,-} > \tau_i^{A,-}$). Using Prop.~\ref{prop:set}, we deduce that $T_{\text{ID}^+_k}^+ >  0$ and thus no vertices in $P_i^{A,-} \cap P_{\text{ID}^+_k}^+$ can belong to $\supscr{P_i}{$A$,pd,$+$}$, and no vertex on the constructed path between $k$ and $P_i^{\mathcal{ID},-}$ belongs to $\text{Proh}_i(t_0^+)$. Since $T_i^+>\tau_i^{A,+}> \frac{1}{s_i}d_{P^{A,-}_i}(k,P^{\mathcal{ID},-}_i)$, $\text{Proh}_i(\cdot)$ remains unchanged over the interval $(t_0, t_0+\frac{1}{s_i}d_{P^{A,-}_i}(k,P^{\mathcal{ID},-}_i)]$.
\end{proof}

\bibliographystyle{plain}
\bibliography{alias,FB,Main,main1,New}
\end{document}